\documentclass[10pt,journal,compsoc]{IEEEtran}

\ifCLASSOPTIONcompsoc
\usepackage[nocompress]{cite}
\else
\usepackage{cite}
\fi
\ifCLASSINFOpdf
\else
\fi
\usepackage{algorithm,makecell,algorithmic,amsbsy,amsmath,amssymb,epsfig,bbm,mathrsfs,multirow,amsthm}
\usepackage{array,multirow,graphicx}
\usepackage[english]{babel}
\usepackage[justification=centering]{caption}
\usepackage{url}
\usepackage{threeparttable}
\usepackage{epstopdf}
\usepackage{subfigure}

\newtheorem{assumption} {Assumption}

\newtheorem{theorem}{\textbf{\textsc{Theorem}}}

\begin{document}
\bibliographystyle{IEEE2}

\title{On Cyber Risk Management of Blockchain Networks: A Game Theoretic Approach}

\author{Shaohan Feng,~\IEEEmembership{Student Member,~IEEE}, Wenbo Wang,~\IEEEmembership{Member,~IEEE}, Zehui Xiong,~\IEEEmembership{Student Member,~IEEE}, Dusit Niyato,~\IEEEmembership{Fellow,~IEEE}, Ping Wang,~\IEEEmembership{Senior Member,~IEEE} and Shaun Shuxun Wang}

\IEEEtitleabstractindextext{
\begin{abstract}
Open-access blockchains based on proof-of-work protocols have gained tremendous popularity for their capabilities of providing decentralized tamper-proof ledgers and platforms for data-driven autonomous organization. Nevertheless, the proof-of-work based consensus protocols are vulnerable to cyber-attacks such as double-spending. In this paper, we propose a novel approach of cyber risk management for blockchain-based service. In particular, we adopt the cyber-insurance as an economic tool for neutralizing cyber risks due to attacks in blockchain networks. We consider a blockchain service market, which is composed of the infrastructure provider, the blockchain provider, the cyber-insurer, and the users. The blockchain provider purchases from the infrastructure provider, e.g., a cloud, the computing resources to maintain the blockchain consensus, and then offers blockchain services to the users. The blockchain provider strategizes its investment in the infrastructure and the service price charged to the users, in order to improve the security of the blockchain and thus optimize its profit. Meanwhile, the blockchain provider also purchases a cyber-insurance from the cyber-insurer to protect itself from the potential damage due to the attacks. In return, the cyber-insurer adjusts the insurance premium according to the perceived risk level of the blockchain service. Based on the assumption of rationality for the market entities, we model the interaction among the blockchain provider, the users, and the cyber-insurer as a two-level Stackelberg game. Namely, the blockchain provider and the cyber-insurer lead to set their pricing/investment strategies, and then the users follow to determine their demand of the blockchain service. Specifically, we consider the scenario of double-spending attacks and provide a series of analytical results about the Stackelberg equilibrium in the market game.
\end{abstract}

}

\maketitle

\IEEEdisplaynontitleabstractindextext

\IEEEpeerreviewmaketitle

\IEEEraisesectionheading{\section{Introduction}
\label{sec:introduction}}

\IEEEPARstart{I}{n} the past few years,  blockchain technologies have attracted tremendous attention from both industry and academia for distributively providing the irreversible, tamper-evident database of tokenized asset transactions~\cite{dinh2017untangling}. Furthermore, with the smart contracts~\cite{bartoletti2017empirical} enabled on top of the decentralized consensus~\cite{dinh2017untangling}, blockchains are envisaged to be the ``game changer'' in various areas ranging from Peer-to-Peer (P2P) resource allocation/trading, e.g., distributed cloud storage~\cite{7966965}, to financial services, e.g., digital identity management~\cite{7961146} and online markets for crowdsourcing services~\cite{8332496}. Although with the advantages such as open access, disintermediation, and pure self-organization, open-access/permissionless blockchains rely on the condition of honest majority to guarantee the data integrity and service security, especially when the Nakamoto consensus protocol based on proof-of-work (PoW) is adopted~\cite{dinh2017untangling}. Since permissionless blockchain networks admit no identity control, they can be vulnerable to a series of insider attacks by malicious consensus nodes~\cite{conti2017survey}. Among different types of attacks, double spending~\cite{Karame:2012:DFP:2382196.2382292} is the most fundamental one and can be executed through various attacks such as goldfinger attacks, netsplit attacks and brute-force attacks~\cite{conti2017survey}. In brief, a double-spending attacker attempts to simultaneously spend the same set of blockchain tokens in two different transactions. This can be performed by first persuading part of the network and the transaction receiver to confirm one transaction, and then persuading the majority of the network to override that transaction with a conflicting transaction spending the same set of tokens. In other words, double-spending attacks are executed through intentional blockchain forking. Due to the factors such as randomness in solving the PoW puzzles~\cite{dinh2017untangling} and information propagation delay, the malicious nodes, i.e., attackers, only need to hold a certain level of PoW computing power to succeed with a high probability in the double-spending attacks. Note here that although the double-spending attack is initially devised for Bitcoin, the attack is also applicable to other blockchain-based resource trading services and systems, for example, energy trading~\cite{8234700}, plug-in hybrid electric vehicle (PHEV) charging credit management~\cite{8306865}, wireless spectrum trading~\cite{7943523}, bandwidth exchange in community networks~\cite{8311658} and cache storage trading~\cite{wenbo}.

Although a few approaches, e.g.,~\cite{Muhammad2018Countering}, have been introduced in blockchains to deter and prevent attacks, due to the inherent characteristics of openness, the PoW-based permissionless blockchain networks may not be completely secure. This critically hinders the broader adoption of permissionless blockchains, especially in business services that require high-level service security. Along with the studies on the improvement of blockchain protocols, blockchain service providers are also looking for an alternative means of cyber-risk management. Recently, cyber-insurance has been recognized as a promising approach to efficiently manage the cyber risks by transferring them to insurers~\cite{pal2013way}. Similar to the traditional insurance, the customer of a cyber-insurance product, i.e., a policyholder, is insured once it settles the contract with the insurer by paying a premium. If attacks happen and the damage is within the coverage of the insurance policy, the insurer will pay the claim to the customer accordingly. To date, a number of cyber-insurance products have been made available in the market~\cite{MAROTTA201735}. According to the types of target systems, these products can be categorized into specific groups designed for service providers such as ISPs and clouds, single mobile/work stations, networks of devices and dedicated cyber-physical/industrial systems.

\begin{figure}[!]
 \centering
 \includegraphics[width=0.45\textwidth,trim=80 15 90 15,clip]{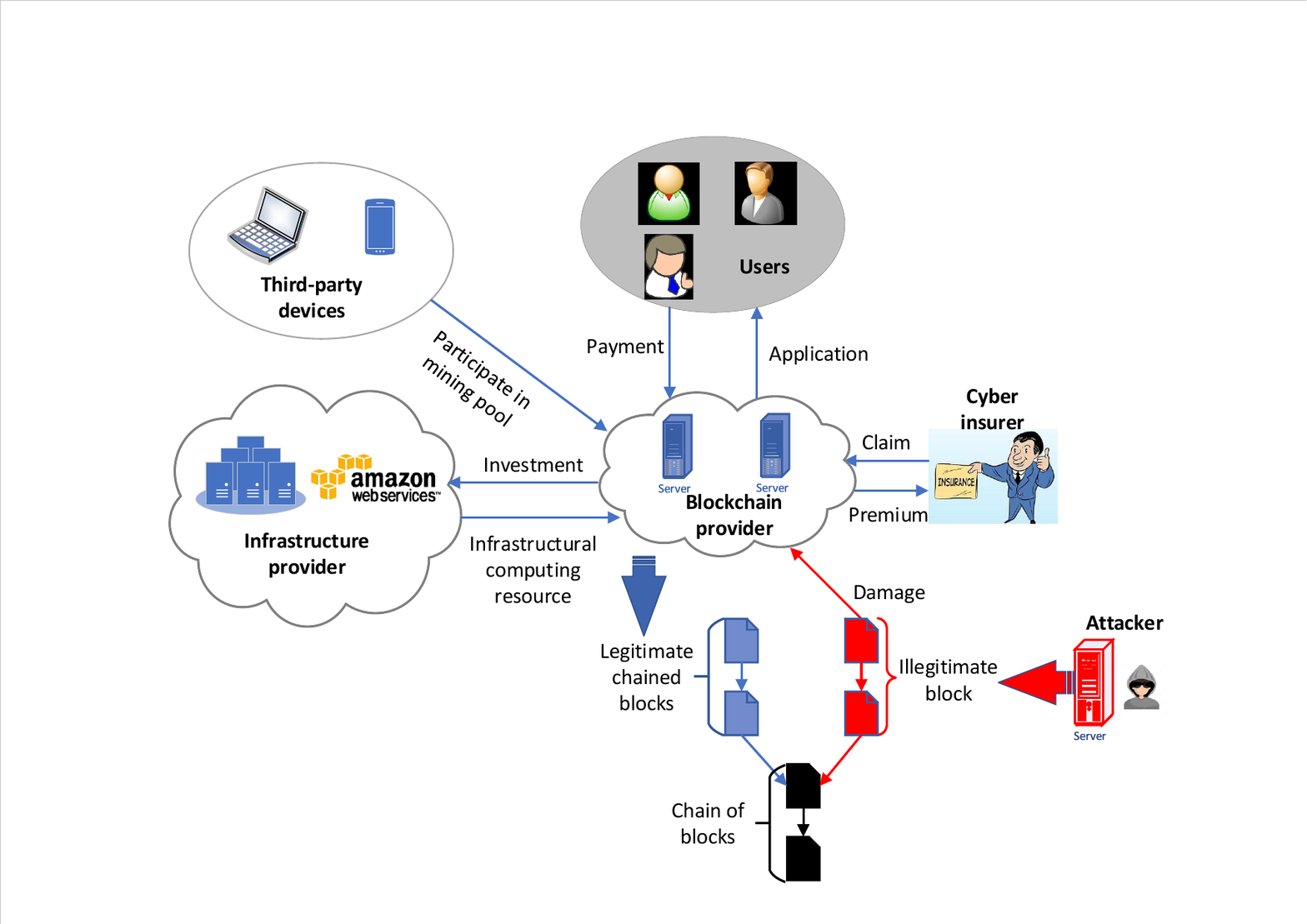}
 \caption{An overview of the blockchain service market.}
 \label{fig:system_model}
\end{figure}

In this paper, we introduce a novel approach of jointly providing the risk management and security enhancement to the blockchain users and providers against attacks through the means of the cyber-insurance. As in other insurable cyber-systems, the market design of cyber-insurance for blockchain networks also has to address a few important issues. Firstly, from the cyber-insurer's perspective, the scope and policy of the cyber-insurance have to be clearly defined in regard to what kind of attacks to be covered and how to quantify the risk, the possible damage and thus the insurance premium. Secondly, alongside the reactive risk transfer with the cyber-insurance, rational blockchain providers also have to consider the proactive strategy in security improvement and thus balance the investment in the infrastructure and in the cyber-insurance.

To answer these questions, we consider a PoW-based blockchain service market under the threat of double-spending attacks (see Fig.~\ref{fig:system_model}). The market is composed of three entities, namely, the users, the blockchain provider, and the cyber-insurer. The users subscribe to a service, e.g., P2P energy trading for smart grids, which is implemented on top of the blockchain provided by the blockchain provider. We consider that the blockchain provider is composed of a group of individual honest machines which are responsible for maintaining the data consensus in the framework of PoW-based permissionless blockchains. The blockchain provider purchases the computing resource from cloud-based infrastructure providers\footnote{For example, Amazon AWS provides the blockchain infrastructures through its partner ecosystem~\cite{AWS2018}.} or deploys more computing power internally for maintaining the network consensus. Meanwhile, in order to lower the cost on infrastructure, third-party devices are encouraged to join the decentralized consensus process by dedicating their individual computing resources into the network. Working as a single blockchain provider, the group of consensus machines make profit by charging the users with the transaction processing fees and block mining fees~\cite{dinh2017untangling}. To neutralize the economic/financial loss incurred by double-spending attacks, the blockchain provider purchases the insurance from the cyber-insurer, which adopts an adjustable premium pricing strategy according to its perceived risk level of the blockchain.

We propose a two-stage Stackelberg game model to analyze the dynamics of the considered market. On the upper stage of the game, the blockchain provider and the cyber-insurer lead to adopt their best-response strategies for profit maximization. On the lower stage, the users adjust their service demands according to the cost and the security level of the blockchain. More specifically, the major contributions of this paper are summarized as follows:
\begin{enumerate}
\item We formulate the mechanisms of blockchain service pricing, blockchain infrastructure investment and cyber-insurance premium adjustment as a joint market equilibrium problem. We model the interactions among the three parties in the market as a two-level Stackelberg game. We provide and prove a few important theoretical discoveries regarding the properties of the equilibrium in the market game.
\item We incorporate the social externality~\cite{gong2017social,xiong2017economic} among blockchain users in our end-user utility model. Also, by modeling the impact of computing power on the blockchain security, we consider the blockchain provider's strategy to incorporate both dimensions of infrastructure investment and insurance spendings. Furthermore, by adopting the concept of ``risk-adjusted premium'', we mathematically capture the impact of attack probability on premium pricing from the cyber-insurer's perspective.
\item We conduct extensive evaluation to assess the performance of the three parties with their equilibrium strategies at different levels of the attacker-controlled computing power.
\end{enumerate}

The proposed market framework introduces a novel incentive-compatible business ecosystem, where the blockchain service users benefit from enjoying more resilient services, and both the blockchain provider and the cyber-insurer are able to gain more profits. In fact, the potential of cyber-insurance for blockchain, Bitcoin specifically, has been perceived in the market. For example, Petra Insurance Brokers (www.insurewithpetra.com) introduces the concept of insurance for Bitcoin transactions. Likewise, BitCoin Financial Group unveils ``BitSecure'', which is a Bitcoin theft insurance policy (www.bitcoinfinancialgroup.com). This product covers both external hacking and employee theft. Moreover, more insurance products are emerging in which a few of them focus on general blockchain services. Therefore, our proposed concept and framework of cyber insurance for blockchain services has a clear and direct practical implication. Moreover, from the perspective of market design, the blockchain services are typically highly customized for a certain environment and application. Thus, we consider the mutual trading between single insurer and single provider. The model of a many-to-many market for multiple insurers and providers is likely to be intractable and needs further development, and hence it is beyond the scope of this paper. Finally, the model developed in this paper can be readily extended to blockchains with the emerging consensus protocols based on the generalized proof of concepts~\cite{dinh2017untangling}, where investment in other resources, e.g., stakes, is needed to prevent attacks.

The rest of the paper is organized as follows. Section~\ref{sec:related_works} presents the related work. Section~\ref{sec:system} describes the preliminaries about blockchains and cyber-insurance, the system model, and the formulation of Stackelberg game. Section~\ref{sec:equilibrium_analysis} investigates the existence and uniqueness of the equilibrium in the proposed Stackelberg game under practical assumptions. Section~\ref{sec:performance} presents the numerical performance evaluation. Section~\ref{sec:conclusion} concludes the paper.

\section{Related Works}
\label{sec:related_works}

The permissionless blockchain network was originally conceptualized in the famous grassroot cryptocurrency project ``Bitcoin'' as a decentralized database for tamper-proof recording and trusted timestamping for the transactional data between P2P users. Permissionless blockchains have been widely recognized for the superb consensus scalability, the tamper-evidence data organization and the capability of supporting the distributed, general-purpose virtual machines~\cite{dinh2017untangling,7423672}. For this reason, in recent years there has seen a plethora of emerging application based on blockchains such as Internet finance and property digitization~\cite{zhao2016overview}, self-organization for Internet of Things~\cite{christidis2016blockchains} and other nonfinancial applications, e.g., notary documents and anti-counterfeit solutions~\cite{wang2016maturity}.

Permissionless blockchain networks creatively solve the problems of replicated consensus in open-access networks by introducing the financial incentive and cryptical zero-knowledge proof into the consensus process~\cite{dinh2017untangling}. More specifically, any node in the blockchain network is allowed to issue digitally signed transactions to other nodes by ``broadcasting'' the transactions in a gossip manner over the P2P links between the nodes. The consensus nodes pack up an arbitrary subset of unapproved transactions into a cryptographically protected data structure, i.e., the ``block'', and rival with each other in a block publishing race (also known as block mining) to acquire the block mining reward~\cite{Garay2015}. Following the Nakamoto protocol for blockchain state maintenance~\cite{dinh2017untangling}, the consensus nodes only accept the longest chain among all the locally observable candidates of the blockchain state as their canonical view of the blockchain. To determine the winner of a block mining race, the consensus nodes have to perform an exhaustive search for the solutions of the crypto-puzzles built upon their proposed blocks. In other words, the nodes that acquire more computing power will have higher probability of winning the race and hence the more power of controlling the blockchain state~\cite{dinh2017untangling}. Therefore, the honest consensus nodes have to secure a sufficiently large amount of computing power to guarantee the well-being of the blockchain services, e.g., data integrity~\cite{Garay2015}. Grabbing more computing power, on the other hand, is also from which the malicious nodes start to breach the blockchain networks.

Cyber-insurance, in the meanwhile, has been recognized as an innovative tool to manage the cyber risks and alleviate the damage of cyber-attacks for the insured customers~\cite{feng2018joint}. Cyber-insurance provides the coverage on losses and liabilities from network/information security breaches. This greatly incentivizes the security investments by cyber-physical systems. However, compared with classical insurance, cyber-insurance introduces a number of unique issues. For example, due to the interdependence of security systems or lack of statistical data, it is difficult to assess the systems' vulnerability and hence hard to estimate the risk transferred to the cyber-insurer~\cite{MAROTTA201735}. For the transaction-oriented cyber services built upon permissionless blokchains, the designer of cyber-insurance also faces the similar issues. Nevertheless, recent studies on the mechanisms of double-spending attacks~\cite{eyal2016bitcoin,rosenfeld2014analysis,grunspan2017double} have shed light upon the possible approaches in analytically assessing the risks of this fundamental threat on the blockchain systems. For example, in~\cite{eyal2016bitcoin},
the authors proposed a new protocol which requires the consensus nodes, i.e., block miners, to confirm transactions only if the inputs of the transactions have not been spent, hence preventing users from double-spending their funds. Based on the characteristics of intentional forking in double-spending attacks, the authors in~\cite{rosenfeld2014analysis,grunspan2017double} introduced the methods to estimate the probability of successful double-spending attacks by analyzing the investments in computing resource of both the blockchain maintainers and attackers. With these studies, it is now possible to estimate the probability of successful double spending and evaluate the potential risks transferred to the cyber-insurers.

Furthermore, under the condition that the probability distribution of risk can be estimated, the authors in~\cite{wang1995insurance} proposed a risk-adjusted premium for pricing risks based on the Proportional Hazard (PH) transform, namely, a power transform in the decumulative distribution function of the risks. Since the PH transform satisfies the elementary principles of assigning premiums, i.e., scale-invariance and translation-invariance, it has been adopted in a number of works concerning the premium determination. In addition, a class of premium functions which are comonotonicity additive and stochastic dominance preservative were studied in~\cite{wang1996premium}. Therein, the premium determination method based on the PH transform was generalized with an axiomatic approach, and the principles such as the absolute deviation principle~\cite{denneberg1990premium} were thoroughly studied and compared.

However, to the best of our knowledge, it still remains an open field of research to employ the cyber-insurance for alleviating the damage of double-spending attack for the blockchain providers. The aforementioned works inspire a vision of quantifying the risk of double-spending attacks based on the computing resources owned by the blockchain maintainers and attackers. Then, the corresponding insurance premium can be determined under the framework based on the PH transform.  By transferring the risks caused by double spending to the cyber-insurers, the cyber-physical services residing on the blockchains can be much more robust. This is also the major objective of our studies.

\section{System Description and Game Formulation}
\label{sec:system}

In this section, we first introduce the model of successful attack probability for double spending and the concept of risk-adjusted premium. After formulating the utilities of the blockchain provider, the users, and the cyber-insurer, we investigate the problem of users' service demand, service pricing and infrastructure, i.e., computing power, investment by the blockchain provider and premium pricing by the cyber-insurer jointly as a hierarchical market game.

The cyber-insurance for blockchain-based service under our consideration works as follows. The blockchain provider firsts pay the premium determined by the cyber-insurer. If the double-spending attack happens, which is detected by the provider or users, the provider files the claims to the cyber-insurer. The cyber-insurer verifies the claim and makes the payment to the provider to compensate for the damage of the double spending. The methods to detect and verify the double spending attack are available and can be adopted, e.g., by deploying an observer as in~\cite{karame2012double}, and hence they are not the focus of this paper.

\subsection{Preliminaries}
\label{subsec:preliminaries}

\subsubsection{Successful Attack Probability}
\label{subsubsec:attack_probability}

We consider that the honest consensus nodes work jointly as a single blockchain provider and are responsible for maintaining a permissionless, PoW-based blockchain for service provision. Extending the analysis in~\cite{rosenfeld2014analysis} and~\cite{grunspan2017double}, we assume that during a time period of $T$, the total computing resource of the blockchain network measured by the hash rate is fixed as $H$. Following the Nakamoto consensus protocol, every consensus node runs an independent Poisson process for puzzle-solving. The average time for a new block to be mined in the blockchain network is $T_{\rm{0}}$~\cite{dinh2017untangling}. Then, in the time period of $T$, the expected number of blocks being successfully mined in the network is $\frac{T}{T_{\rm{0}}}$. Let $h$ denote the investment in computing resources by the blockchain provider, i.e., the honest nodes, and  $a$ denote the investment in computing resources by the attackers. Then, if the computing efficiency for hash queries are roughly the same, the blockchain provider and the attackers divide the total computing resource $H$ as ${\bar h}H$ and ${\bar a}H$, respectively, where $\bar h = \frac{h}{{a + h}}$ and $\bar a = \frac{a}{{a + h}}$ are the investment ratios. According to the probabilistic model for winning the PoW-based puzzle solving race~\cite{rosenfeld2014analysis}, the number of blocks that are mined by the blockchain provider and waiting for confirmation during $T$ is $ {\frac{T}{{{T_{\rm{0}}}}}\frac{{{\bar{h}}H}}{{H}}} = \frac{T}{{{T_{\rm{0}}}}}{\bar{h}}$. On the other hand, instead of following the Poisson distribution based model, the number of blocks successfully mined by attackers during $T$ can be accurately modeled as a negative binomial variable~\cite{rosenfeld2014analysis}. Therefore, with the investment ratio ${\bar{h}}$, the probability for attackers to succeed in double spending during $T$ can be expressed as follows (see Theorem~1 in~\cite{grunspan2017double}):
\begin{equation}\label{Eq:attack_probability}
\footnotesize{{\rm{P}}\left( {\overline h } \right) = {I_{4\left( {1 - \bar h} \right)\bar h}}\left( {\frac{T}{{{T_{\rm{0}}}}}\bar h,\frac{1}{2}} \right), {\bar{h}} \ge \frac{1}{2},}
\end{equation}
where ${I_w}\left( {u,v} \right)$ is the regularized incomplete Beta function:
\begin{equation}
\footnotesize{{I_w}\left( {u,v} \right) = \frac{{\Gamma \left( {u + v} \right)}}{{\Gamma \left( u \right)\Gamma \left( v \right)}}\int\nolimits_0^w {{t^{u - 1}}{{\left( {1 - t} \right)}^{v - 1}}} {\rm{d}}t}
\end{equation}
with $\Gamma\left(\cdot\right)$ being the gamma function. The model of exponential decay in~(\ref{Eq:attack_probability}) is discovered in~\cite{nakamoto2009bitcoin} and proved in~\cite{grunspan2017double}.

We consider that the blockchain provider receives payments from the users in the form of transaction fees in a confirmed block. Under double-spending attacks, the blockchain provider has to compensate the loss of the users with a fixed rate for each transaction in the block that is finally overridden. Assume that the number of transactions included in each block is the same, and hence the transaction fee and compensation rate are fixed for each transaction. Let $N_{\rm{T}}$ denote the number of transactions in a block, $r$ denote the block mining reward for each block and $q$ denote the total compensation rate for each block. Then, with the investment ratio ${\bar{h}}$, the blockchain provider's potential loss is $\frac{T}{{{T_{\rm{0}}}}}{\bar{h}} N_{\rm{T}} q$ under the double-spending attack, and the probability of successful attack is:
\begin{equation}\label{Eq:attack_probability_model}
\footnotesize{{\rm{P}}\left( {\bar h} \right) = \left\{ {\begin{aligned}
&{I_{4\left( {1 - \bar h} \right)\bar h}}\left( {\frac{T}{{{T_{\rm{0}}}}}\bar h,\frac{1}{2}} \right),&\bar h \ge \frac{1}{2},\\
&1,&\bar h < \frac{1}{2}.
\end{aligned}} \right.}
\end{equation}
Here, we only focus on the case where the investment ratio of the blockchain provider is no less than ${1}/{2}$. The reason is that when ${\bar{h}}<{1}/{2}$, the probability of successful double spending is always equal to $1$ as shown in~(\ref{Eq:attack_probability_model}), which is trivial and thus not our focus.

\subsubsection{Premium Determination}
\label{subsubsec:premium_determination}

The cyber-insurer offers a cyber-insurance service to the blockchain provider and covers its total loss when the attack happens. In other words, after the blockchain provider buys the cyber-insurance, the risk of double-spending attack will be transferred to the cyber-insurer. By adopting the concept of risk-adjusted premium in~\cite{wang1996premium}, the cyber-insurer dynamically determines the price, i.e., premium, of its cyber-insurance product according to the insurance risk distribution. According to our previous discussion, the cyber-insurer has an insurance risk, i.e., paying the claim of $\frac{T}{{{T_{\rm{0}}}}}{\bar{h}} N_{\rm{T}} q$ with the probability of ${\rm{P}}\left( {\bar h} \right)$ given by (\ref{Eq:attack_probability_model}). Then, the expected loss for the cyber-insurer can be formulated as follows:
\begin{equation}\label{Eq:expected_claim_insurer}
\footnotesize{\begin{aligned}
{{\rm{E}}_{loss}} &= \int\nolimits_{1/2}^1 {\frac{T}{{{T_{\rm{0}}}}}\bar h{N_{\rm{T}}}q{\rm{P}}\left( {\bar h} \right){\rm{d}}\bar h}  = \frac{T}{{{T_{\rm{0}}}}}{N_{\rm{T}}}q\int\nolimits_{1/2}^1 {\bar h{\rm{P}}\left( {\bar h} \right){\rm{d}}\bar h} \\
 &= \frac{T}{{{T_{\rm{0}}}}}{N_{\rm{T}}}q\left[ {\left. {\bar hF\left( {\bar h} \right)} \right|_{\bar h = 1/2}^{\bar h = 1} - \int\nolimits_{1/2}^1 {{\rm{F}}\left( {\bar h} \right){\rm{d}}\bar h} } \right] \\
 &\le \frac{T}{{{T_{\rm{0}}}}}{N_{\rm{T}}}q\left[ {\frac{3}{4} - \int\nolimits_{1/2}^1 {{\rm{F}}\left( {\bar h} \right){\rm{d}}\bar h} } \right]\\
 &= \frac{T}{{{T_{\rm{0}}}}}{N_{\rm{T}}}q\left[ {\int\nolimits_{1/2}^1 {\frac{3}{2}{\rm{d}}\bar h}  - \int\nolimits_{1/2}^1 {{\rm{F}}\left( {\bar h} \right){\rm{d}}\bar h} } \right] \\
 &= \frac{T}{{{T_{\rm{0}}}}}{N_{\rm{T}}}q\int\nolimits_{1/2}^1 {\frac{3}{2} - {\rm{F}}\left( {\bar h} \right){\rm{d}}\bar h} \\
 &= \frac{T}{{{T_{\rm{0}}}}}{N_{\rm{T}}}q\int\nolimits_{1/2}^1 {\left[ {1 - \int\nolimits_{1/2}^{\bar h} {{\rm{P}}\left( \theta  \right){\rm{d}}\theta } } \right]{\rm{d}}\bar h},
\end{aligned}}
\end{equation}
where ${\rm{F}}\left({\bar{h}}\right)$ is the cumulative distribution function for ${\rm{P}}\left({\bar{h}}\right)$\footnote{Note here that we assume that the cyber-insurer is considering the worst case.}.
Based on the formulated distribution of the insurance risk and the concept of risk-adjusted premium, the cyber-insurer can determine the premium as follows:
\begin{equation}\label{Eq:premium_determination}
\footnotesize{\Lambda \left( \gamma  \right) = \frac{T}{{{T_{\rm{0}}}}}{N_{\rm{T}}}q\int\nolimits_{1/2}^1 {\omega \left({1 - \int\nolimits_{1/2}^{\bar h} {{\rm{P}}\left( \theta  \right){\rm{d}}\theta } } , \gamma \right){\rm{d}}\bar h},}
\end{equation}
where $\omega \left(x,\gamma\right)$ is an increasing concave function of $x$ and belongs to the families of elementary transforms given in Section~5 of~\cite{wang1996premium}. Without loss of generality, we adopt the PH transform in our study, i.e., $\omega\left( x,\gamma \right) = {x^{\frac{1}{\gamma }}}$, $\gamma  \ge 1$ in Subsection~5.1 of~\cite{wang1996premium}. Then, the corresponding premium can be expressed as follows:
\begin{equation}\label{Eq:premium_determination_PH}
\footnotesize{\Lambda \left( \gamma  \right) = \frac{T}{{{T_{\rm{0}}}}}{N_{\rm{T}}}q\int\nolimits_{1/2}^1 {\left[ {1 - \int\nolimits_{1/2}^{\bar h} {{\rm{P}}\left( \theta  \right){\rm{d}}\theta } }\right]^{1/\gamma}{\rm{d}}\bar h},}
\end{equation}
where $\gamma$ is the premium coefficient which decides on the insurance policy. Namely, the cyber-insurer adjusts the premium by controlling $\gamma$ according to the insurance risk. It is worth noting that the term ${\left[ {1 - \int\nolimits_{1/2}^{\bar h} {{\rm{P}}\left( \theta  \right){\rm{d}}\theta } } \right]}$ in~(\ref{Eq:premium_determination_PH}) is smaller than $1$. Therefore, the larger $\gamma$ is, the higher the premium $\Lambda\left(\gamma\right)$ will be.

\subsection{System Model}
\label{subsec:model}

\subsubsection{The User's Utility}
\label{subsubsec:users_utility}

We suppose that each user in the blockchain service market has a service demand, which is determined by an intrinsic value $\theta_i$ from the uniform distribution ${\rm{F}}_{\rm{U}}$ over the interval $\left[0, 1\right]$. Here, $\theta_i$ can be interpreted as the probability for user $i$ to buy the blockchain service. We further assume that the intrinsic values of the users are independently distributed. The users also experience social externalities in which the decision of one user can influence the decisions of the other users. Let ${\rm{Pr}}\left[ \text{j buys the service} \right]$ denote the probability that user $j$ subscribes to the service, then, the utility of user $i$ can be expressed as follows~\cite{bloch2013pricing}:
\begin{equation}\label{Eq:utilitu_function}\begin{footnotesize}
{u_i} = {\bar{h}}+ {\theta _i} - {p_i} + \alpha \sum\limits_{j \in {\cal N}} {{g_{ij}}{\rm{Pr}}\left[ \text{j buys the service} \right]}  .\end{footnotesize}
\end{equation}
In (\ref{Eq:utilitu_function}), the first term, i.e., ${\bar{h}}$, is the investment ratio of the blockchain provider. ${\bar{h}}$ represents the positive effect owning to the effort of the blockchain provider in preventing the security breach due to double-spending attacks. According to (\ref{Eq:attack_probability}), the larger ${\bar{h}}$ is, the lower the successful probability of double-spending attacks is, and consequently the less the users will be affected. The third term, i.e., $p_i$, is the price of the service for user $i$\footnote{Note here that the uniform pricing in which all users are charged with the same price is a special case of the discriminative pricing adopted in this paper.}. The forth term of (\ref{Eq:utilitu_function}), i.e., {\footnotesize{$\alpha \sum\limits_{j \in {\cal N}} {{g_{ij}}{\rm{Pr}}\left[ \text{j buys the service} \right]}$}}, represents the positive social externality among the users. $g_{ij}$ is the element of the externality matrix ${\bf{G}}$~\cite{gong2017social},~\cite{xiong2017economic} and represents the level of the social externality (influence) that user $j$ has on user $i$. We assume that $g_{ij} \ne 0$, $\forall i \ne j$ and $g_{ii} = 0$, $\forall i \in {\cal{N}}$\footnote{Note here that our model can also be used to analyze the scenario that the users may act independently. For example, if $g_{ij}=g_{ji}=0$, user $j$ has zero level of the social externality (influence) on user $i$ and vice versa. Hence, users $i$ and $j$ act independently.}. Finally, $\alpha$ is a constant controlling the level of social externality for the entire network.

\subsubsection{Profits of the Blockchain Provider and the Cyber-insurer}
\label{subsubsec:profit}

The goals of the blockchain provider and the cyber-insurer are to maximize their individual profits. Based on our previous discussion, the payoff functions of the blockchain provider can be expressed as follows:
\begin{equation}\label{Eq:profit_of_provider}
\footnotesize{\begin{aligned}
{\Pi _{\rm{p}}}\left( {\bar h,{\bf{p}}} \right) =& \sum\limits_{i \in {\cal{N}}} {{p_i}{x_i}}  - \frac{{a\bar h}}{{1 - \bar h}} + {\bar{h}}\frac{T}{{{T_{\rm{0}}}}}{N_{\rm{T}}}r \\
&- \underbrace {\frac{T}{{{T_{\rm{0}}}}}{N_{\rm{T}}}q\int\nolimits_{1/2}^1 {{{\left[ {1 - \int\nolimits_{1/2}^t {{\rm{P}}\left( \theta  \right){\rm{d}}\theta } } \right]}^{1/\gamma }}{\rm{d}}t} }_{\text{premium}},
\end{aligned}}
\end{equation}
where the first term $\sum\nolimits_{i \in {\cal{N}}} {{p_i}{x_i}}$ is the revenue obtained from the users' payment for the blockchain service and $x_i$ is the probability that user $i$ buys the service defined in (\ref{Eq:x_i}). The second term, i.e., $h \!= \!\frac{{a\bar h}}{{1 \!- \!\bar h}}$, is the blockchain provider's investment in the infrastructure, and we have $h \!= \!\frac{{a\bar h}}{{1 \!- \!\bar h}} \!\Leftrightarrow\! \bar h\!=\!\frac{h}{{a \!+ \!h}}$. The third term, i.e., $\bar h\frac{T}{{{T_{\rm{0}}}}}{N_{\rm{T}}}r$, is the block mining reward received by the blockchain provider for maintaining the service. The last term, i.e., ${\frac{T}{{{T_{\rm{0}}}}}{N_{\rm{T}}}q\int\nolimits_{1/2}^1 {{{\left[ {1 - \int\nolimits_{1/2}^t {{\rm{P}}\left( \theta  \right){\rm{d}}\theta } } \right]}^{1/\gamma }}{\rm{d}}t} }$, is the premium paid by the blockchain provider to the cyber-insurer, and $q$ is the compensation price for one transaction.

On the other hand, the premium paid by the blockchain provider is the revenue of the cyber-insurer. Due to the uncertainty of double-spending attacks, we adopt the expected claim as the cyber-insurer's cost. Then, the cyber-insurer's payoff function can be expressed as:
\begin{equation}\label{Eq:profit_of_insurer}
\footnotesize{\begin{aligned}
{\Pi _{\rm{I}}}\left( \gamma  \right) =& \underbrace {\frac{T}{{{T_{\rm{0}}}}}{N_{\rm{T}}}q\int\nolimits_{1/2}^1 {{{\left[ {1 - \int\nolimits_{1/2}^t {{\rm{P}}\left( \theta  \right){\rm{d}}\theta } } \right]}^{1/\gamma }}{\rm{d}}t} }_{\text{premium}}\\
 &- {\rm{P}}\left( {\bar h} \right)\bar h\frac{T}{{{T_{\rm{0}}}}}{N_{\rm{T}}}q-\sigma({\bar{h}}, \gamma),
\end{aligned}}
\end{equation}
where the first term is the premium paid by the blockchain provider. The second term, i.e., ${\rm{P}}\left( {\bar h} \right) \bar h\frac{T}{{{T_{\rm{0}}}}}{N_{\rm{T}}}q$, is obtained as the product of the successful attack probability and the total claim paid to the blockchain provider for covering its loss. Finally, since the premium increases as the cyber-insurer's decision variable $\gamma$ increases, a rational cyber-insurer will keep $\gamma$ as high as possible to maximum its revenue. However, when the blockchain provider increases its investment ratio, the successful attack probability will decrease. As a result, the blockchain provider will have less incentive to pay an extremely high premium. Therefore, we introduce the last term in (\ref{Eq:profit_of_insurer}) to model the possibly negative impact on the expected payoff of the cyber-insurer as both the values of ${\bar{h}}$ and $\gamma$ increase, which is similar to the concept of the punishment on insurer adopted in~\cite{bloch2008informal}. Here, we define $\sigma \left( {\bar h,\gamma } \right) = {\sigma _1}\left( {\bar h} \right){\sigma _2}\left( \gamma  \right)$, where ${\sigma _1}\left( {\bar h} \right)$ is an increasing convex function of ${\bar{h}}$ with the following properties:
\begin{equation}\label{Eq:sigma_1}
\footnotesize{{\sigma _1}\left( {\bar h} \right)\left\{ {\begin{aligned}
&{ > 0,}&{\bar h > \frac{1}{2}},&\\
&{ = 0,}&{\bar h = \frac{1}{2}},&\\
&{ < 0,}&{\bar h < \frac{1}{2}}.&
\end{aligned}} \right.}
\end{equation}
The conditions in (\ref{Eq:sigma_1}) indicates that with ${\bar{h}} > \frac{1}{2}$, the blockchain provider's effort in investing the computing resource effectively reduces the successful attack probability. Consequently, the probability of the cyber-insurer paying the claim to the blockchain provider is also reduced. Then, keeping the highest premium has a negative effect on the cyber-insurer's payoff, e.g., by hurting its reputation or curbing the incentive for the blockchain provider to buy the insurance. Under the other two conditions in (\ref{Eq:sigma_1}), the lack of enough investment in infrastructure of the blockchain provider will induce higher probability of being successfully attacked, hence leading to a positive effect on the cyber-insurer's payoff due to the higher demand of financial protection. Additionally, $\sigma_2 \left(  \gamma  \right)$ is an increasing convex function of $\gamma$ and $\sigma_2 \left.\left(  \gamma  \right)\right|_{\gamma=1} = 0$. As such, when $\gamma=1$, the expected loss in~(\ref{Eq:expected_claim_insurer}) is equal to the premium in~(\ref{Eq:premium_determination_PH}), and there is no negative effect on the cyber-insurer's reputation. For tractable analysis, we adopt the following model:
\begin{footnotesize}\begin{equation}
  \label{Eq:sigma_model}
\sigma \left( {\bar h,\gamma } \right) = {\sigma _1}\left( {\bar h} \right){\sigma _2}\left( \gamma  \right) = \underbrace {{{\left( {\bar h - \frac{1}{2}} \right)}^3}}_{{\sigma _1}\left( {\bar h} \right)}\underbrace {\left( {\gamma  - 1} \right){\gamma ^\beta }}_{{\sigma _2}\left( \gamma  \right)}, \quad \beta>1.
\end{equation}\end{footnotesize}
It is worth noting that the cyber-insurer's payoff function in (\ref{Eq:profit_of_insurer}) may also adopt other models for $\sigma \left( {\bar h,\gamma } \right)$. The selected model in (\ref{Eq:sigma_model}) has no effect on our subsequent analysis.

\subsection{Stackelberg Game Formulation}
\label{subsec:game_formulation}

Considering the payoff functions of the market entities given in~(\ref{Eq:utilitu_function})-(\ref{Eq:profit_of_insurer}), it is natural to model the interactions in the blockchain service market as a two-stage game. In the first upper stage, the blockchain provider determines the price of the blockchain service, namely, the levels of acceptable transaction fees for each user $p_i$, and its ratio of investment in computing resources ${\bar{h}}$. Meanwhile, the cyber-insurer decides on the premium coefficient $\gamma$ by considering the insurance risk transferred from the blockchain provider. In the second lower stage, each user determines whether it will buy the blockchain service or not based on the prices and the investment ratio set by the blockchain provider. Accordingly, the interactions among the blockchain provider, the cyber-insurer and the users are formulated as a two-leader-multi-follower Stackelberg game. Specifically, the mutual interaction between the blockchain provider and cyber-insurer forms a noncooperative two-player leader-level subgame for achieving the equilibrium of the service price, the investment ratio, and the cyber-insurance policy. Then, the interaction among a number of users forms the follower-level noncooperative subgame for determining the service demand from the blockchain provider\footnote{The reason that the users are in the same level, i.e., lower-level, is that they have the same set of information and make decisions simultaneously. This is different from the blockchain provider that usually invests in the infrastructure and buys the cyber-insurance first to improve the security level of its blockchain-based service and then provides the service to the users. The cyber-insurer sells the cyber-insurance to the blockchain provider and forms a noncooperative relationship with the blockchain provider. As such, the cyber-insurer and blockchain provider make their decisions simultaneously and hence before the users. As a result, the blockchain provider and cyber-insurer are considered to be the leaders and their problems are defined in the upper-level.}. The Stackelberg game can be formally defined as follows.

\begin{enumerate}
\item \emph{User-level noncooperative subgame}: Given the fixed investment ratio ${\bar{h}}$ as well as the price vectors \begin{footnotesize}${\bf{p}}=\left[p_1, p_2, \ldots, p_{\left|{\cal{N}}\right|} \right]^\top$\end{footnotesize}, the user-level (follower) noncooperative subgame is defined by a four-tuple \begin{footnotesize}${\cal{G}}_{\rm{u}}=\left\{{\cal{N}}, {\bf{x}}, {\cal{X}}, {\bf{u}}\right\}$\end{footnotesize}, where
    \begin{itemize}
    \item \begin{footnotesize}${\cal{N}}$\end{footnotesize} is the set of active users;
    \item \begin{footnotesize}${\cal{X}}= \left\{ {\left. {{{\left[ {{x_1},{x_2}, \ldots {x_{\left| {\cal{N}} \right|}}} \right]}^\top}} \right|{x_i} \in \left[ {0,1} \right],i \in {\cal{N}}} \right\}$ $ \subset {\mathbb{R}}^{\left| {\cal{N}} \right|}$\end{footnotesize} defines the domain of ${\bf{x}}$ as an M-polyhedron;
    \item \begin{footnotesize}${\bf{x}} = {\left[ {{x_1},{x_2}, \ldots {x_{\left| {\cal{N}} \right|}}} \right]^\top}$\end{footnotesize} is the vector of the users' decision variables, where $x_i$ is the service demand of user~$i$ and \begin{footnotesize}${\bf{x}}\in{\cal{X}}$\end{footnotesize};
    \item \begin{footnotesize}${\bf{u}} = \left[u_1, u_2, \ldots, u_{\left|{\cal{N}}\right|}\right]^{\top}$\end{footnotesize} is the vector of the users' utilities with the given strategy \begin{footnotesize}${\bf{x}}$\end{footnotesize}, where \begin{footnotesize}$\forall i \in {\cal{N}}$, and $u_i$\end{footnotesize} is given in~(\ref{Eq:utilitu_function}).
    \end{itemize}
\item \emph{Leader-level noncooperative subgame}: Assume that the users' demand ${\bf{x}}$ has been found to be a parametric equilibrium as a function, i.e., mapping, of the leaders' strategies. Then, the blockchain provider and the cyber-insurer form a noncooperative game as a five-tuple \begin{footnotesize}${{\cal{G}}_{\rm{L}}} = \left\{ {{{\left[ {{{\bf{p}}^ \top },\bar h} \right]}^ \top },{{\cal{D}}_{\rm{P}}},\gamma ,{{\cal{D}}_{\rm{I}}},{\bf{\Pi}} } \right\}$\end{footnotesize}, where
    \begin{itemize}
    \item \begin{footnotesize}$\left[ {{{\bf{p}}^ \top },\bar h} \right]^\top=\left[ {{p_1},{p_2}, \ldots ,	{p_{\left| {\cal{N}} \right|}}, {\bar{h}}} \right]^\top$\end{footnotesize} is the strategy vector of the service prices and the investment ratio set by the blockchain provider with $p_i > 0$, $\forall i \in {\cal{N}}$ and ${\bar{h}} \in \left[\frac{1}{2}, 1\right)$;
    \item \begin{footnotesize}${{\cal{D}}_{\rm{P}}}=\left\{ \left. {{{\left[ {{{\bf{p}}^ \top },\bar h} \right]}^ \top }} \right|{p^{\rm{u}}} \ge {p_i} > 0,\forall i \in {\cal{N}}, \right.$ $ \left. \bar h \in \left[ {\frac{1}{2},1} \right) \right\}$\end{footnotesize} is the domain of the prices and investment ratio of the blockchain provider, where $p^{\rm{u}}$ is the upper bound of the price $p_i$ imposed by the government or market regulators;
    \item $\gamma$ is the cyber-insurer's premium coefficient for premium determination;
    \item \begin{footnotesize}${{\cal{D}}_{\rm{I}}}=\left\{\left. \gamma \right| \gamma^{\rm{u}} \ge \gamma > 1 \right\}$\end{footnotesize} defines the domain of $\gamma$, where $\gamma^{\rm{u}}$ is the upper bound of $\gamma$ imposed by the government or regulators;
    \item \begin{footnotesize}${\bf{\Pi}}=\left[{\Pi _{\rm{p}}}, {\Pi _{\rm{I}}}\right]^\top$\end{footnotesize} is the profit vector for the blockchain provider and the cyber-insurer.
    \end{itemize}
\end{enumerate}

\section{Game Equilibrium Analysis}
\label{sec:equilibrium_analysis}

Based on the formulation of the Stackelberg game in Section~\ref{subsec:game_formulation}, we are ready to analyze the market equilibrium using backward induction. We first obtain the Nash Equilibrium (NE) of the user-level noncooperative subgame ${\cal{G}}_{\rm{u}}$ by characterizing a system of interdependent demands. We provide the sufficient conditions for the existence and uniqueness of the NE in the user-level noncooperative subgame ${\cal{G}}_{\rm{u}}$ by solving the bounded linear complementarity problem of the subgame~\cite{bloch2013pricing}. Then, we substitute the parametric NE of ${\cal{G}}_{\rm{u}}$ into the leader-level noncooperative subgame ${\cal{G}}_{\rm{L}}$. Under a reasonable assumption, we show that the Jacobian matrix constructed from the payoff functions of each player in the leader-level noncooperative subgame is negative definite. Hence, we prove that the Stackelberg Equilibrium (SE) of the market game exists and is unique.

\subsection{Equilibrium Analysis for User-Level Noncooperative Subgame}
\label{subsec:solution_users_level_noncooperative_game}

Intuitively, user $i$ only buys the service when it has a positive payoff, namely $u_i>0$. This indicates that there exists a threshold ${\tilde \theta }_i$ for the intrinsic value ${\theta }_i$ in (\ref{Eq:utilitu_function}), such that user $i$ will buy the service only when $\theta_i > {\tilde \theta }_i$. ${\tilde \theta }_i$ can be obtained by setting $u_i=0$:
\begin{equation}
\footnotesize{\begin{aligned}
0 &= {\bar{h}}+ {\tilde \theta }_i - {p_i} + \alpha \sum\limits_{j \in {\cal N}} {{g_{ij}}{\rm{Pr}}\left[ \text{j buys the service} \right]} \\
 &= {\bar{h}}+ {\tilde \theta }_i - {p_i} + \alpha \sum\limits_{j \in {\cal N}} {{g_{ij}}{\rm{Pr}}\left[ \theta_j > {\tilde \theta }_j \right]}\\
 &= {\bar{h}}+ {\tilde \theta }_i - {p_i} + \alpha \sum\limits_{j \in {\cal N}} {{g_{ij}}\left[ {1 - {\rm{F}}_{\rm{U}}\left( {{{\tilde \theta }_j}} \right)} \right]}\\
  \Leftrightarrow {{\tilde \theta }_i} &= {p_i} - {\bar{h}} - \alpha \sum\limits_{j \in {\cal N}} {{g_{ij}}\left[ {1 - {\rm{F}}_{\rm{U}}\left( {{{\tilde \theta }_j}} \right)} \right]},
\end{aligned}}
\end{equation}
where \begin{footnotesize}${1 - {\rm{F}}_{\rm{U}}\left( {{{\tilde \theta }_j}} \right)}$\end{footnotesize} denotes the probability that user $j$ draws a valuation above the threshold ${{\tilde \theta }_j}$. For ease of exposition, let \begin{footnotesize}$x_i = {1 - {\rm{F}}_{\rm{U}}\left( {{{\tilde \theta }_i}} \right)}$\end{footnotesize} denote the probability that user $i$ buys the service, and $x_i$ can be further expressed as follows:
\begin{footnotesize}\begin{equation}\label{Eq:x_i}
\begin{aligned}
{x_i} &= 1 - {\rm{F}}_{\rm{U}}\left( {{{\tilde \theta }_i}} \right) = 1 - {\rm{F}}_{\rm{U}}\left( {{p_i} - {\bar{h}} - \alpha \sum\limits_{j {\in {\cal{N}}}} {{g_{ij}}\left[ {1 - {\rm{F}}_{\rm{U}}\left( {{{\tilde \theta }_j}} \right)} \right]} } \right)\\
 &= 1 - {p_i} + {\bar{h}} + \alpha \sum\limits_{j {\in {\cal{N}}}} {{g_{ij}}x_j}.
\end{aligned}
\end{equation}\end{footnotesize}
From~(\ref{Eq:x_i}), we can characterize a system of the users' interdependent demands as follows:
\begin{footnotesize}\begin{equation}\label{Eq:demand_system_1}
{x_i} = \left\{ {\begin{array}{*{20}{c}}
{0,}&{\text{if $1 - {p_i} + {\bar{h}} + \alpha \sum\limits_{j {\in {\cal{N}}}} {{g_{ij}}x_j}   < 0$}},\\
{1,}&{ \text{if $1 - {p_i} + {\bar{h}} + \alpha \sum\limits_{j {\in {\cal{N}}}} {{g_{ij}}x_j}    > 1$}},\\
{1 - {p_i} + {\bar{h}} + \alpha \sum\limits_{j {\in {\cal{N}}}} {{g_{ij}}x_j}  ,}&{\text{otherwise}}.
\end{array}} \right.
\end{equation}\end{footnotesize}
After subtracting each condition in (\ref{Eq:demand_system_1}) by the value of its corresponding $x_i$, we convert (\ref{Eq:demand_system_1}) into the following form:
\begin{footnotesize}\begin{equation}\label{Eq:demand_system_2}
\begin{aligned}
&{x_i} =\\
&\left\{ {\begin{aligned}
{0,}\quad{\text{if $1 - p_i + {\bar{h}} + \alpha \sum\limits_{j \in {\cal{N}}} {g_{ij}{x_j}}  - {x_i} < 0 - \underbrace {{x_i}}_{ = 0} = 0$,}}\\
{1,}\quad{\text{ if $1 - p_i + {\bar{h}} + \alpha \sum\limits_{j \in {\cal{N}}} {g_{ij}{x_j}}  - {x_i}  > 1 - \underbrace {{x_i}}_{ = 1} = 0$,}}\\
{1 - {p_i} + {\bar{h}} + \alpha \sum\limits_{j \in {\cal{N}}} {{g_{ij}}{x_j}} ,}\quad{\text{if $1 - p_i + {\bar{h}} + \alpha \sum\limits_{j \in {\cal{N}}} {g_{ij}{x_j}}  - {x_i}  = 0$}.}
\end{aligned}} \right.
\end{aligned}
\end{equation}\end{footnotesize}
Furthermore, (\ref{Eq:demand_system_2}) can be rewritten into the following matrix form:
\begin{equation}\label{Eq:demand_system_3}
\footnotesize{\begin{aligned}
&{x_i} =\\
&\left\{ {\begin{aligned}
0,\qquad\qquad&\quad\text{if ${{\left\{ {\left( {1 + {\bar{h}}} \right){\bf{1}} - {\bf{p}} - \left( {{\bf{I}} - \alpha {\bf{G}}} \right){\bf{x}}} \right\}}_i} < 0$,}\\
1,\qquad\qquad&\quad\text{if ${{\left\{ {\left( {1 + {\bar{h}}} \right){\bf{1}} - {\bf{p}} - \left( {{\bf{I}} - \alpha {\bf{G}}} \right){\bf{x}}} \right\}}_i} > 0$,}\\
{{\left\{ {\left( {1 + {\bar{h}}} \right){\bf{1}} - {\bf{p}} + \alpha {\bf{G}}{\bf{x}}} \right\}}_i},&\quad\text{if ${{\left\{ {\left( {1 + {\bar{h}}} \right){\bf{1}} - {\bf{p}} - \left( {{\bf{I}} - \alpha {\bf{G}}} \right){\bf{x}}} \right\}}_i} = 0$,}
\end{aligned}} \right.
\end{aligned}}
\end{equation}
where \begin{footnotesize}${\bf{I}}$\end{footnotesize} is the identity matrix, \begin{footnotesize}${\bf{1}}=\left[1,1,\ldots,1\right]^\top \in {\mathbb{R}}^{\left|{\mathcal{N}}\right| \times 1}$ and $\left\{\cdot\right\}_i$\end{footnotesize} represents the $i$-th entry of a vector. With (\ref{Eq:demand_system_3}), we are ready to investigate the properties of the NE in the user-level noncooperative subgame.

\begin{assumption}\label{Ass:existent_and_uniqueness_for_user_game}
\begin{footnotesize}$\alpha \rho \left({\bf{G}}\right) <1$\end{footnotesize}, where \begin{footnotesize}$\rho \left(\cdot\right)$\end{footnotesize} is the spectral norm of a matrix.
\end{assumption}
The physical meaning of Assumption~\ref{Ass:existent_and_uniqueness_for_user_game} is illustrated in Fig.~\ref{fig:assumption}, where the black area is the feasible area of $\alpha \rho \left({\bf{G}}\right)$. If $\alpha \rho \left({\bf{G}}\right)$ exceeds the feasible area, the social externality will be too strong such that every user will buy the service, which is impossible in reality.

\begin{figure}[!]
 \centering
 \includegraphics[width=0.4\textwidth,trim=110 5 120 30,clip]{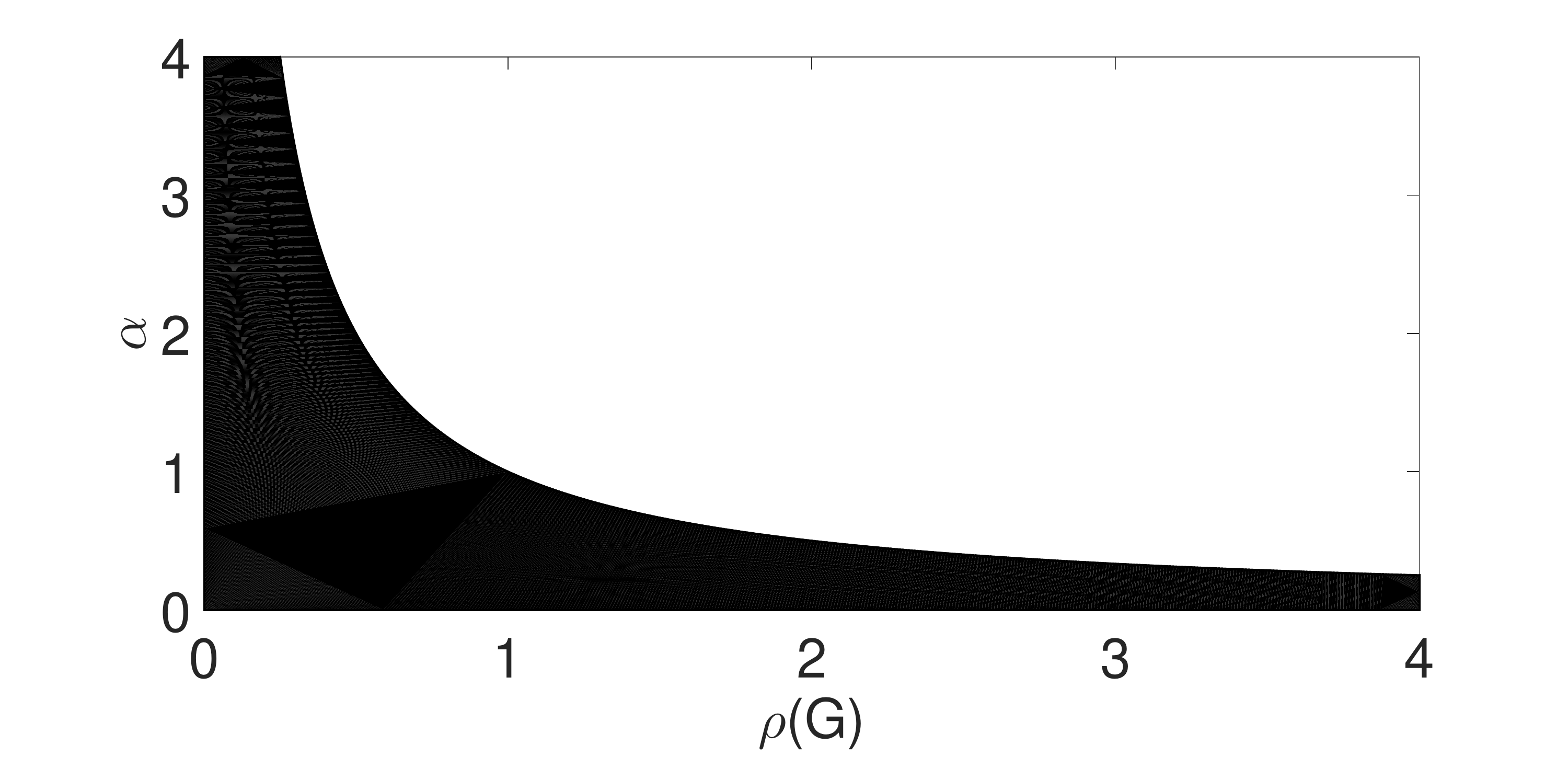}
 \caption{An illustration of Assumption~\ref{Ass:existent_and_uniqueness_for_user_game}.}
 \label{fig:assumption}
\end{figure}

\begin{theorem}\label{Th:existent_and_uniqueness_for_user_game}
The user-level noncooperative subgame ${\cal{G}}_{\rm{u}}$ admits a unique NE under Assumption~\ref{Ass:existent_and_uniqueness_for_user_game}.
\end{theorem}

\begin{proof}
(\ref{Eq:demand_system_3}) is defined as a bounded linear complementarity problem in~\cite{bloch2013pricing}. It is a linear instance of the general mixed complementarity problems discussed in~\cite{simsek2005uniqueness}. As discussed in~\cite{simsek2005uniqueness}, the linear instance of a complementarity problem admits a unique solution, i.e., the NE, to the user-level noncooperative subgame \begin{footnotesize}${\cal{G}}_{\rm{u}}$\end{footnotesize}, if \begin{footnotesize}$\left({\bf{I}}-\alpha {\bf{G}}\right)$\end{footnotesize} is a P-matrix\footnote{A matrix $A$ is a P-matrix if all its principal minors are positive.}.

Since \begin{footnotesize}$\alpha \rho \left({\bf{G}}\right) <1$\end{footnotesize} under Assumption~\ref{Ass:existent_and_uniqueness_for_user_game}, all the eigenvalues of the matrix \begin{footnotesize}$\alpha {\bf{G}}$\end{footnotesize} belong to \begin{footnotesize}$\left(0, 1\right)$\end{footnotesize} and hence all the eigenvalues of the matrix \begin{footnotesize}$\left({\bf{I}}-\alpha {\bf{G}}\right)$\end{footnotesize} belong to \begin{footnotesize}$\left(0, 1\right)$\end{footnotesize}. Therefore, \begin{footnotesize}$\left({\bf{I}}-\alpha {\bf{G}}\right)$\end{footnotesize} is a non-singular M-matrix\footnote{A matrix $A$ is a non-singular M-matrix if $A = I - B$ for a positive matrix $B$ with largest eigenvalue $\rho \left(B\right) < 1$.}.

Based on Theorem 6.2.3 in~\cite{berman1994nonnegative}, ``any
non-singular M-matrix is a P-matrix'', \begin{footnotesize}$\left({\bf{I}}-\alpha {\bf{G}}\right)$\end{footnotesize} is a P-matrix and then there exists a unique NE in the user-level noncooperative subgame \begin{footnotesize}${\cal{G}}_{\rm{u}}$\end{footnotesize}. Thereby, the proof is completed.
\end{proof}

According to the system of interdependent demands (\ref{Eq:demand_system_3}), the set of users ${\cal{N}}$ can be partitioned into three subsets, i.e., \begin{footnotesize}${\cal{S}}_0$, ${\cal{S}}_1$, and ${\cal{S}}$\end{footnotesize} as follows:
\begin{itemize}
\item \begin{footnotesize}${{\cal{S}}_0} = \left\{ {i\left| {{{\left\{ {\left( {1 + {\bar{h}}} \right){\bf{1}} - {\bf{p}} - \left( {{\bf{I}} - \alpha {\bf{G}}} \right){\bf{x}}} \right\}}_i} < 0,\forall i \in {\cal{N}}} \right.} \right\}$\end{footnotesize} is the set of users which will not buy the blockchain service,

\item \begin{footnotesize}${{\cal{S}}_1} = \left\{ {i\left| {{{\left\{ {\left( {1 + {\bar{h}}} \right){\bf{1}} - {\bf{p}} - \left( {{\bf{I}} - \alpha {\bf{G}}} \right){\bf{x}}} \right\}}_i} > 0,\forall i \in {\cal{N}}} \right.} \right\}$\end{footnotesize} is the set of users which surely will buy the blockchain service,

\item \begin{footnotesize}${\cal{S}} = {\cal{N}} \backslash \left( {{S_0} \cup {S_1}} \right) = \{ {i\vert {\{ {( 1 + {\bar{h}} ){\bf{1}} - {\bf{p}} - ( {{\bf{I}} - \alpha {\bf{G}}} ){\bf{x}}} \}}_i} = 0,\forall i \in {\cal{N}} \}$\end{footnotesize} is the set of users which buy the blockchain service with a probability over $\left[0, 1\right]$.
\end{itemize}
Then, we obtain the following theorem:
\begin{theorem}\label{Th:set_belong_for_user_game}
The users in the user-level noncooperative subgame \begin{footnotesize}${\cal{G}}_{\rm{u}}$\end{footnotesize} only belong to \begin{footnotesize}$\cal{S}$\end{footnotesize}, i.e., \begin{footnotesize}${\cal{S}}={\cal{N}}$\end{footnotesize}, \begin{footnotesize}${\cal{S}}_0=\emptyset$\end{footnotesize}, and \begin{footnotesize}${\cal{S}}_1=\emptyset$\end{footnotesize}. Given the service price vector \begin{footnotesize}${\bf{p}}$\end{footnotesize}, the optimal solution to the system of interdependent demands (\ref{Eq:demand_system_3}) is
\begin{footnotesize}\begin{equation}
{{\bf{x}}^*} = {\left( {{\bf{I}} - \alpha {\bf{G}}} \right)^{ - 1}}\left[ {\left( {1 + {\bar{h}}} \right){\bf{1}}} - {\bf{p}} \right].
\end{equation}\end{footnotesize}
\end{theorem}

\begin{proof}
We denote the optimal price by \begin{footnotesize}${\bf{p}}^*$\end{footnotesize} and the optimal demand by \begin{footnotesize}${\bf{x}}^*$\end{footnotesize}. Then, we show that \begin{footnotesize}${\cal{S}}_0=\emptyset$\end{footnotesize} and \begin{footnotesize}${\cal{S}}_1=\emptyset$\end{footnotesize}.

\begin{enumerate}
\item \begin{footnotesize}${\cal{S}}_0=\emptyset$\end{footnotesize}: We first assume that \begin{footnotesize}${\cal{S}}_0 \ne \emptyset$\end{footnotesize}. This means that \begin{footnotesize}$\exists i\in{\cal{N}}$, such that ${x_i}^*=0$\end{footnotesize} and \begin{footnotesize}$\{(1 + {\bar{h}} ){\bf{1}} - {\bf{p}}^* - ({\bf{I}}- \alpha {\bf{G}} ){\bf{x}}^*\}_{i}<0$\end{footnotesize}. Because \begin{footnotesize}${{{\left\{ {\left( {1 + {\bar{h}}} \right){\bf{1}} - {\bf{p}} - \left( {{\bf{I}} - \alpha {\bf{G}}} \right){\bf{x}}} \right\}}_i}}$ $=1 + {\bar{h}} + \alpha \sum\limits_{j \in {\cal{N}}} {g_{ij}{x_j}} - x_i -p_i$\end{footnotesize} is continuous on $p_i$ and \begin{footnotesize}$1 + {\bar{h}} + \alpha \sum\limits_{j \in {\cal{N}}} {g_{ij}{x_j}^*} > 0$\end{footnotesize}, there exists a \begin{footnotesize}${p_i}'$\end{footnotesize} where \begin{footnotesize}${p_i}' < {p_i}^*$\end{footnotesize} such that \begin{footnotesize}$1 + {\bar{h}} + \alpha \sum\limits_{j \in {\cal{N}}} {g_{ij}{x_j}^*}-{p_i}' > 0$\end{footnotesize} and correspondingly \begin{footnotesize}${x_i}' > 0$\end{footnotesize}. This indicates that when the service price charged to user $i$ decreases from \begin{footnotesize}${p_i}^*$\end{footnotesize} to \begin{footnotesize}${p_i}'$\end{footnotesize}, user $i$ has an incentive to increase its demand from \begin{footnotesize}${x_i}^* = 0$\end{footnotesize} to \begin{footnotesize}${x_i}' > 0$\end{footnotesize}. Consequently, the revenue of the blockchain provider will increase since \begin{footnotesize}${p_i}^*{x_i}^* = 0$\end{footnotesize} while \begin{footnotesize}${p_i}'{x_i}' > 0$\end{footnotesize}. Therefore, \begin{footnotesize}${x_i}^*$\end{footnotesize} and \begin{footnotesize}${p_i}^*$\end{footnotesize} cannot be the optimal demand and price for user $i$, respectively. Hence, \begin{footnotesize}$i \notin {\cal{S}}_0$, $\forall i \in {\cal{N}}$\end{footnotesize} and \begin{footnotesize}${\cal{S}}_0=\emptyset$\end{footnotesize}.

\item \begin{footnotesize}${\cal{S}}_1=\emptyset$\end{footnotesize}: Similarly, we assume that \begin{footnotesize}${\cal{S}}_1 \ne \emptyset$\end{footnotesize}. This means that \begin{footnotesize}$\exists l \in {\cal{N}}$, ${x_l}^*=1$\end{footnotesize} and \begin{footnotesize}${{{\left\{ {\left( {1 + {\bar{h}}} \right){\bf{1}} - {\bf{p}}^* - \left( {{\bf{I}} - \alpha {\bf{G}}} \right){\bf{x}}^*} \right\}}_l} > 0}$\end{footnotesize}. Since \begin{footnotesize}${{{\left\{ {\left( {1 + {\bar{h}}} \right){\bf{1}} - {\bf{p}} - \left( {{\bf{I}} - \alpha {\bf{G}}} \right){\bf{x}}} \right\}}_l}}$ $=1 + {\bar{h}} + \alpha \sum\limits_{j \in {\cal{N}}} {g_{ij}{x_j}} - x_l -p_l$\end{footnotesize} is continuous on \begin{footnotesize}$p_l$\end{footnotesize}, there exists an \begin{footnotesize}$\epsilon$\end{footnotesize} where \begin{footnotesize}$\epsilon > 0$\end{footnotesize} such that \begin{footnotesize}$1 + {\bar{h}} + \alpha \sum\limits_{j \in {\cal{N}}} {g_{ij}{x_j}^*}- \left( {p_l}^* + \epsilon \right) > 0$\end{footnotesize} and \begin{footnotesize}${x_l}^*=1$\end{footnotesize}. Let \begin{footnotesize}$\epsilon = \bar h + \alpha \sum\limits_{j \in N} {{g_{ij}}{x_j}^*}  - {p_l}^*$\end{footnotesize}, we have \begin{footnotesize}$1 + {\bar{h}} + \alpha \sum\limits_{j \in {\cal{N}}} {g_{ij}{x_j}^*}- \left( {p_l}^* + \epsilon \right) = 0$\end{footnotesize}, and here \begin{footnotesize}${x_l}^* = 1 - \left({p_l}^*+\epsilon\right) + {\bar{h}} + \sum\limits_{j {\in {\cal{N}}}} {{g_{lj}}{x_j}^*}=  1$\end{footnotesize}. This means that even if the service price of user $l$ increases from \begin{footnotesize}${p_l}^*$\end{footnotesize} to \begin{footnotesize}$\left({p_l}^*+\epsilon\right)$\end{footnotesize}, the demand of user $l$ is still equal to $1$ while the profit of the blockchain provider has been increased from \begin{footnotesize}${p_l}^*{x_l}^*={p_l}^*$\end{footnotesize} to \begin{footnotesize}$\left({p_l}^*+\epsilon\right){x_l}^*={p_l}^*+\epsilon$\end{footnotesize}. Moreover, since
    \begin{equation}\label{Eq:epsilon_max}
    \footnotesize{\begin{aligned}
    &{\left\{ {\left( {1 + \bar h} \right){\bf{1}} - {{\bf{p}}^*} - \left( {{\bf{I}} - \alpha {\bf{G}}} \right){{\bf{x}}^*}} \right\}_l} - {\epsilon }\\
    =& 1 + \bar h - \left( {{p_l}^* + {\epsilon }} \right) - \underbrace {{x_l}^*}_1 + \alpha \sum\limits_{j \in {\cal{N}}} {{g_{lj}}{x_j}^*} = 0,
    \end{aligned}}
    \end{equation}
    user $l$ belongs to \begin{footnotesize}${\cal{S}}$\end{footnotesize} instead of \begin{footnotesize}${\cal{S}}_1$\end{footnotesize}. Therefore, \begin{footnotesize}${p_l}^*$\end{footnotesize} and \begin{footnotesize}${x_l}^*$\end{footnotesize} are not the optimal price and demand for user $l$, respectively. Hence, \begin{footnotesize}$l \notin {\cal{S}}_1$, $\forall l \in {\cal{N}}$ and ${\cal{S}}_1 = \emptyset$\end{footnotesize}.
\end{enumerate}
To conclude, given any price vector \begin{footnotesize}${\bf{p}}$\end{footnotesize} of the blockchain service, the condition
\begin{equation}
\footnotesize{{\left\{ {\left( {1 + {\bar{h}}} \right){\bf{1}} - {\bf{p}} - \left( {{\bf{I}} - \alpha {\bf{G}}} \right){{\bf{x}}^*}} \right\}_i} = 0}
\end{equation}
will be satisfied for all user $i \in {\cal{N}}$. Therefore,
\begin{footnotesize}\begin{equation}
\left( {1 + {\bar{h}}} \right){\bf{1}} - {\bf{p}} - \left( {{\bf{I}} - \alpha {\bf{G}}} \right){{\bf{x}}^*} = {\bf{0}},
\end{equation}\end{footnotesize}
\begin{footnotesize}\begin{equation}\label{Eq:NE_to_user_game_1}
\left( {{\bf{I}} - \alpha {\bf{G}}} \right){{\bf{x}}^*} = \left( {1 + {\bar{h}}} \right){\bf{1}} - {\bf{p}}.
\end{equation}\end{footnotesize}
Since we have already shown that the matrix \begin{footnotesize}$\left( {{\bf{I}} - \alpha {\bf{G}}} \right)$\end{footnotesize} is a non-singular matrix in Theorem~\ref{Th:existent_and_uniqueness_for_user_game}, the inverse of \begin{footnotesize}$\left( {{\bf{I}} - \alpha {\bf{G}}} \right)$\end{footnotesize} exists. Multiply both sides of~(\ref{Eq:NE_to_user_game_1}) by \begin{footnotesize}$\left( {{\bf{I}} - \alpha {\bf{G}}} \right)^{-1}$\end{footnotesize}, then, the optimal solution to the system of interdependent demands, or equivalently, the NE to the user-level noncooperative subgame \begin{footnotesize}${\cal{G}}_{\rm{u}}$\end{footnotesize}, is
\begin{footnotesize}\begin{equation}\label{Eq:optimal_solution_to_the_user_level_game}
{{\bf{x}}^*} = {\left( {{\bf{I}} - \alpha {\bf{G}}} \right)^{ - 1}}\left[ {\left( {1 + {\bar{h}}} \right){\bf{1}}} - {\bf{p}} \right].
\end{equation}\end{footnotesize}
The proof is completed.
\end{proof}

\subsection{Equilibrium Analysis for Leader-Level Noncooperative Subgame}
\label{subsec:solution_to_the_leader_level_game}

After deriving the equilibrium demand of the users, we investigate the leader-level noncooperative subgame ${\cal{G}}_{\rm{L}}$ for the blockchain provider and the cyber-insurer. At the NE, no player can increase its profit by choosing a different strategy provided that the other player' strategy is unchanged~\cite{osborne2004introduction}. In what follows, we first prove that an NE exists in the leader-level noncooperative subgame ${\cal{G}}_{\rm{L}}$. Then, we prove that this NE is unique.

Substituting the optimal demand of the users derived in (\ref{Eq:optimal_solution_to_the_user_level_game}) into the profit function of the blockchain provider, we can rewrite~(\ref{Eq:profit_of_provider}) into a matrix form as follows:
\begin{footnotesize}\begin{equation}\label{Eq:profit_of_provider_1}
\begin{aligned}
{{\Pi_{\rm{p}}}}\left( {\bar h,{\bf{p}}} \right)=& {{\bf{p}}^ \top }{\left( {{\bf{I}} - \alpha {\bf{G}}} \right)^{ - 1}}\left[ {\left( {1 + \bar h} \right){\bf{1}} - {\bf{p}}} \right] - \frac{{a\bar h}}{{1 - \bar h}} + {\bar{h}}\frac{T}{{{T_{\rm{0}}}}}{N_{\rm{T}}}r \\
&- \frac{T}{{{T_{\rm{0}}}}}{N_{\rm{T}}}q\int\nolimits_{1/2}^1 {{{\left[ {1 - \int\nolimits_{1/2}^t {{\rm{P}}\left( \theta  \right){\rm{d}}\theta } } \right]}^{1/\gamma }}{\rm{d}}t}.
\end{aligned}
\end{equation}\end{footnotesize}

\begin{figure*}
\begin{footnotesize}\begin{equation}\label{Eq:first_partial_drivative_of_profit}
  \left\{
\begin{aligned}
\frac{{\partial {{\Pi_{\rm{p}}}}}}{{\partial \bar h}} =& {{\bf{p}}^ \top }{\left( {{\bf{I}} - \alpha {\bf{G}}} \right)^{ - 1}}{\bf{1}} - \frac{a}{{{{\left( {1 - \bar h} \right)}^2}}} + \frac{1}{\gamma }\frac{T}{{{T_{\rm{0}}}}}{N_{\rm{T}}}r,\\
\frac{{\partial {\Pi_{\rm{P}}}}}{{\partial {\bf{p}}}} = & {\left( {{\bf{I}} - \alpha {\bf{G}}} \right)^{ - 1}}\left[ {\left( {1 + \bar h} \right){\bf{1}} - {\bf{p}}} \right],\\
\frac{{\partial {\Pi _{\rm{I}}}}}{{\partial \gamma }} =&  - \frac{1}{{{\gamma ^2}}}\frac{T}{{{T_{\rm{0}}}}}{N_{\rm{T}}}q\int\nolimits_{1/2}^1 {{{\left[ {1 - \int\nolimits_{1/2}^t {{\rm{P}}\left( \theta  \right){\rm{d}}\theta } } \right]}^{1/\gamma }}\ln \left[ {1 - \int\nolimits_{1/2}^t {{\rm{P}}\left( \theta  \right){\rm{d}}\theta } } \right]{\rm{d}}t}
- {\left( {\overline h  - \frac{1}{2}} \right)^3}\left[ {\left( {\beta  + 1} \right){\gamma ^\beta } - \beta {\gamma ^{\beta  - 1}}} \right],
\end{aligned}\right.
\end{equation}\end{footnotesize}
\begin{footnotesize}\begin{equation}\label{Eq:second_partial_drivative_of_profit}
  \left\{
\begin{aligned}
\frac{{{\partial ^2}{\Pi _{\rm{P}}}}}{{\partial {{\bf{p}}^2}}} &=  - {\left( {{\bf{I}} - \alpha {\bf{G}}} \right)^{ - 1}},\\
\frac{{{\partial ^2}{\Pi _{\rm{P}}}}}{{\partial {\bf{p}}\partial \bar h}} &= {\left( {{\bf{I}} - \alpha {\bf{G}}} \right)^{ - 1}}{\bf{1}},\\
\frac{{{\partial ^2}{\Pi _{\rm{P}}}}}{{\partial {\bf{p}}\partial \gamma }} &= {\left( {\frac{{{\partial ^2}{\Pi _{\rm{I}}}}}{{\partial \gamma \partial {\bf{p}}}}} \right)^ \top } = {\bf{0}},\\
\frac{{{\partial ^2}{\Pi _{\rm{P}}}}}{{\partial \bar h\partial {\bf{p}}}} &= {{\bf{1}}^ \top }{\left( {{\bf{I}} - \alpha {\bf{G}}} \right)^{ - 1}},\\
\frac{{{\partial ^2}{\Pi _{\rm{P}}}}}{{\partial {{\bar h}^2}}} &=  - \frac{{2a}}{{{{\left( {1 - \bar h} \right)}^3}}},\\
\frac{{{\partial ^2}{\Pi _{\rm{P}}}}}{{\partial \bar h\partial \gamma }} &= 0,
\end{aligned}\right.\qquad \qquad
\left\{
\begin{aligned}
\frac{{{\partial ^2}{\Pi _{\rm{I}}}}}{{\partial \gamma \partial \bar h}} =&  - 3{\left( {\overline h  - \frac{1}{2}} \right)^2}\left[ {\left( {\beta  + 1} \right){\gamma ^\beta } - \beta {\gamma ^{\beta  - 1}}} \right],\\
\frac{{{\partial ^2}{\Pi _{\rm{I}}}}}{{\partial {\gamma ^2}}} =& \frac{2}{{{\gamma ^3}}}\frac{T}{{{T_{\rm{0}}}}}{N_{\rm{T}}}q\int\nolimits_{1/2}^1 {{{\left[ {1 - \int\nolimits_{1/2}^t {{\rm{P}}\left( \theta  \right){\rm{d}}\theta } } \right]}^{1/\gamma }}\ln \left[ {1 - \int\nolimits_{1/2}^t {{\rm{P}}\left( \theta  \right){\rm{d}}\theta } } \right]{\rm{d}}t} \\
 &+ \frac{1}{{{\gamma ^4}}}\frac{T}{{{T_{\rm{0}}}}}{N_{\rm{T}}}q\int\nolimits_{1/2}^1 {{{\left[ {1 - \int\nolimits_{1/2}^t {{\rm{P}}\left( \theta  \right){\rm{d}}\theta } } \right]}^{1/\gamma }}{{\ln }^2}\left[ {1 - \int\nolimits_{1/2}^t {{\rm{P}}\left( \theta  \right){\rm{d}}\theta } } \right]{\rm{d}}t} \\
 &- {\left( {\overline h  - \frac{1}{2}} \right)^3}\left[ {\beta \left( {\beta  + 1} \right){\gamma ^{\beta  - 1}} - \beta \left( {\beta  - 1} \right){\gamma ^{\beta  - 2}}} \right].
\end{aligned}\right.
\end{equation}\end{footnotesize}
\end{figure*}

The first-order partial derivative of the profit of blockchain provider as well as that of the cyber-insurer are shown in~(\ref{Eq:first_partial_drivative_of_profit}), and the second-order partial derivatives of the blockchain provider and cyber-insurer are shown
in~(\ref{Eq:second_partial_drivative_of_profit}). Then, we can obtain the following theorem regarding the NE of the leader-level noncooperative subgame.
\begin{theorem}\label{Th:existence_for_leader_game}
There exists at least one NE in the leader-level noncooperative subgame ${\cal{G}}_{\rm{L}}$ if and only if \begin{footnotesize}$a >  \frac{1}{8}{{\bf{1}}^ \top }{{\left( {{\bf{I}} - \alpha {\bf{G}}} \right)}^{ - 1}}{\bf{1}}$\end{footnotesize}. Then, the Stackbelberg equilibrium of the market game exists.
\end{theorem}

\begin{proof}
Please refer to Appendix~\ref{subsec:existence_leader_game} for the proof.
\end{proof}

Next, we show the uniqueness of the NE in the leader-level noncooperative subgame in Theorem~\ref{Th:uniqueness_for_leader_game}, and hence the Stackbelberg equilibrium is unique.
\begin{theorem}\label{Th:uniqueness_for_leader_game}
The NE in the leader-level noncooperative subgame ${\cal{G}}_{\rm{L}}$ is unique if and only if \begin{footnotesize}$a >  \frac{{9{{\left( {\beta  + 1} \right)}^2}{{\left( {{\gamma ^u}} \right)}^{\beta  + 1}}}}{{128\beta }}$\end{footnotesize} and hence the Stackbelberg equilibrium is unique.
\end{theorem}

\begin{proof}
Please refer to Appendix~\ref{subsec:uniqueness_leader_game} for the proof.
\end{proof}

\subsection{Equilibrium Searching for Leader-Level Noncooperative Subgame}
\label{sub:sec:algorithm}
\begin{algorithm}[t]
  \begin{footnotesize}
 \caption{Iterative best-response for searching leader noncooperative subgame NE}
 \begin{algorithmic}[1]
 \REQUIRE
 Select any feasible initial strategies as $\mathbf{s}_{\textrm{I}}(0)=\gamma(0)$, $\mathbf{s}_{\textrm{P}}(0)=[\mathbf{p}(0), {\overline{h}}(0)]^{\top}$ and set $t=0$.
 \WHILE {$\mathbf{p}(0), \overline{h}(0), \gamma(0)$ do not satisfy the termination condition}
  \FORALL {$k=\{\textrm{P}, \textrm{I}\}$}
  \STATE Set the adversary joint strategies as
  \begin{equation}
   \label{eq_set_adversary}
    \mathbf{s}_k(t+1)=\arg\max_{\mathbf{s}_k}\Pi_k(\mathbf{s}_k, \mathbf{s}_{-k}(t)) \textrm{  s.t. } \mathbf{s}_k\in\mathcal{D}_k,
  \end{equation}
  where $\mathbf{s}_{-k}$ represents the adversary's strategy.
  \ENDFOR
  \STATE Set $t\leftarrow t+1$.
 \ENDWHILE
 \end{algorithmic}
 \label{alg_iterative}
\end{footnotesize}
\end{algorithm}

Since in our previous discussion we have provided the closed-form NE to the user-level noncooperative subgame ${\cal{G}}_{\rm{u}}$ in (\ref{Eq:optimal_solution_to_the_user_level_game}), we only have to focus on the derivation of the NE in the leader-level noncooperative subgame ${\cal{G}}_{\rm{L}}$. Given (\ref{Eq:optimal_solution_to_the_user_level_game}), the search for the SE is reduced to the search of the NE of a two-player noncooperative game. By Theorem~\ref{Th:existence_for_leader_game}, we know that ${\cal{G}}_{\rm{L}}$ is a  concave game with the convex and compact strategy space (see also the proof to Theorem~\ref{Th:existence_for_leader_game}). By Theorem~\ref{Th:uniqueness_for_leader_game}, we know that ${\cal{G}}_{\rm{L}}$ admits a unique NE when the condition given therein is satisfied. This suggests the use of the iterative best response to solve for the NE in ${\cal{G}}_{\rm{L}}$.  The iterative best-response algorithm is described in Algorithm~\ref{alg_iterative}, and its convergence property is guaranteed by Theorem~\ref{thm_convergence}.
\begin{theorem}\label{thm_convergence}
If the condition in Theorem~\ref{Th:uniqueness_for_leader_game} is satisfied, Algorithm~\ref{alg_iterative} converges to the unique SE from anywhere of the strategy domains of the blockchain provider and cyber-insurer.
\end{theorem}

\begin{proof}
With the concavity of the payoff functions $\Pi_{\textrm{P}}$ and $\Pi_{\textrm{I}}$ proved in Appendix~\ref{subsec:existence_leader_game} and the negative definite Jacobian matrix $\mathbf{J}$ proved in Appendix~\ref{subsec:uniqueness_leader_game}, Theorem~\ref{thm_convergence} immediately follows Theorem 10 in~\cite{scutari2014real}.
\end{proof}

\section{Performance Evaluation}
\label{sec:performance}

In this section, we conduct extensive numerical simulations to evaluate the performance of the market entities at the equilibrium in each stage. We consider a group of $\left|{\cal{N}}\right|$ users in the blockchain service market. The off-diagonal elements of social externality matrix ${\bf{G}}$, i.e., $g_{ij}$, $\forall i \ne j$, is generated following the uniform distribution over the interval of $\left[0, 10\right]$. The domain of definition for $\alpha$ is set as $\left[5\times10^{-4}, 8 \times 10^{-4}\right]$ according to the parameter setting in~\cite{bloch2013pricing}. The other default coefficients are given as follows: $\beta=10$, $p^{\rm{u}}=1$, $\gamma^{\rm{u}}=2$, $T=100$, $T_{\rm{0}}=10$, $N_{\rm{T}}=100$, $r=10$, $q=10$ and $a=100$. Note that the price that we present in the figures of simulation results is the mean value of the discriminatory prices. The triple integral in the premium-related term in the profit functions of the blockchain provider and the cyber-insurer, i.e.,~(\ref{Eq:profit_of_provider}) and~(\ref{Eq:profit_of_insurer}), is calculated using the method of rectangular integral with $100$ as the number of intervals. Note that our proposed concept of cyber-insurer and blockchain service is the first in the literature. There is no similar work with which we can compare on a reasonably fair basis. For example, the authors in~\cite{zhang2017bi} formulated a bilevel game to investigate the interactions among the attackers, users, and cyber-insurer in computer networks. However, since we incorporate the specific and unique feature of the blockchain technology, i.e., successful attack probability in~(\ref{Eq:attack_probability}), reasonable comparison is therefore not applicable.

\subsection{Numerical Results}
\label{subsec:numerical_result}

\subsubsection{Demonstration of best response and NE}
\label{subsubsec:ne_and_best_response}

\begin{figure}[!]
\centering
\begin{minipage}{4.4cm}
\centering
\includegraphics[width=1\textwidth,trim=5 0 15 5,clip]{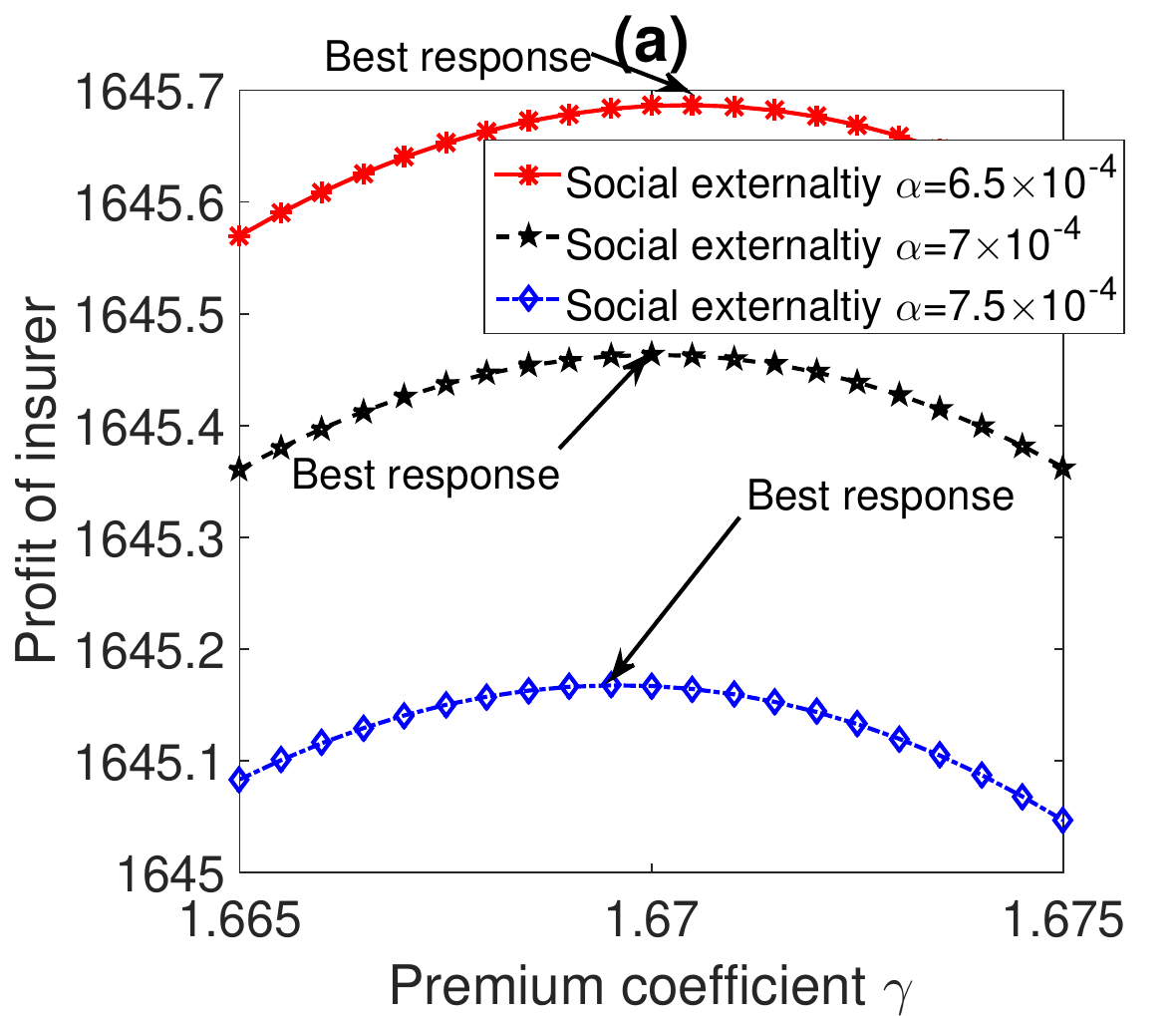}
\end{minipage}
\begin{minipage}{4.4cm}
\centering
\includegraphics[width=1\textwidth,trim=5 0 15 5,clip]{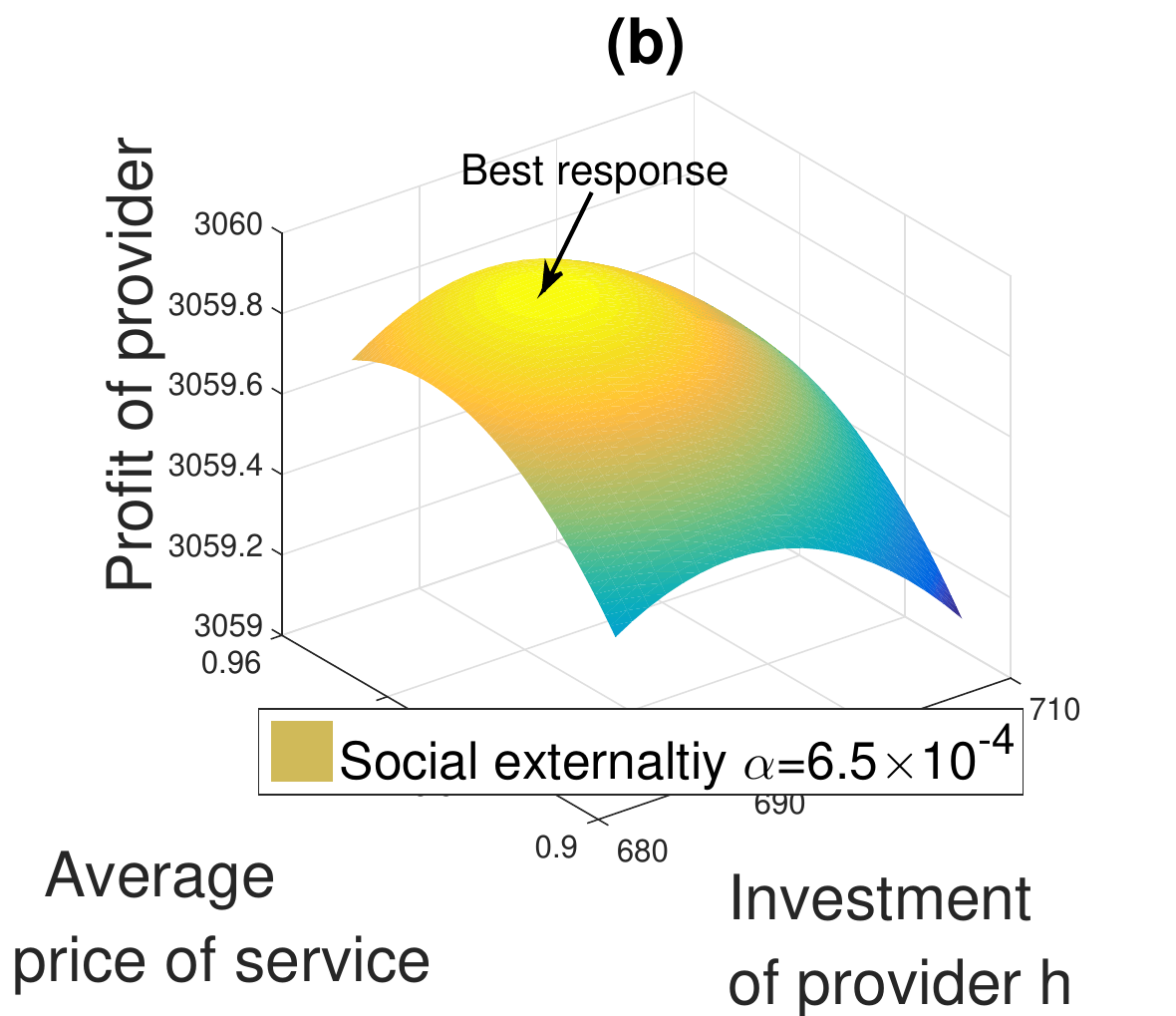}
\end{minipage}
\begin{minipage}{4.4cm}
\centering
\includegraphics[width=1\textwidth,trim=5 0 15 5,clip]{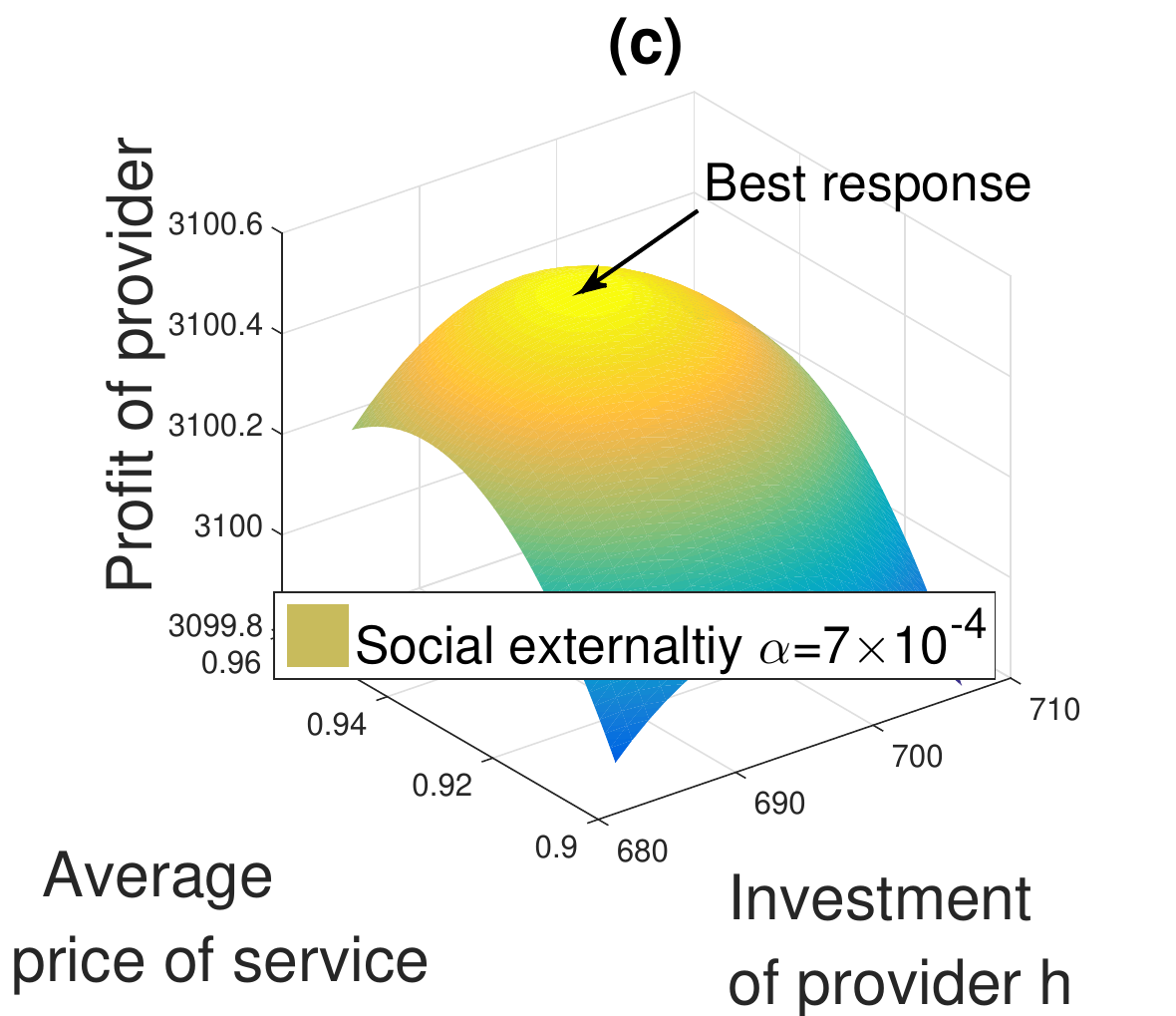}
\end{minipage}
\begin{minipage}{4.4cm}
\centering
\includegraphics[width=1\textwidth,trim=5 0 15 5,clip]{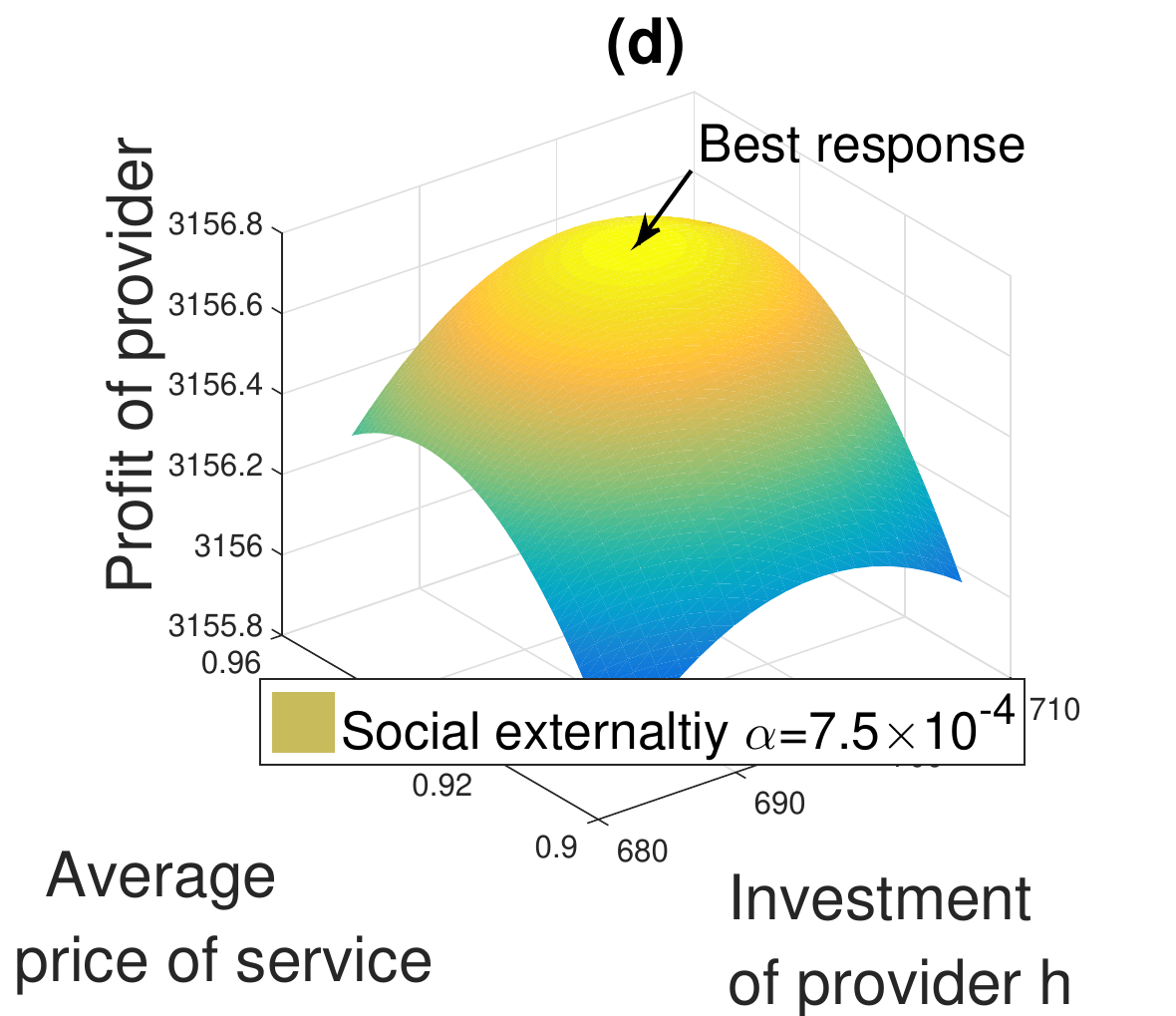}
\end{minipage}
\caption{Best response}
\label{fig:best_response}
\end{figure}

Figure~\ref{fig:best_response} demonstrates the NE obtained from interative best responses with the given simulation parameters. In Fig.~\ref{fig:best_response}(a), the profit of the cyber-insurer changes with the different premium charged to the blockchain provider. According to our previous discussion, the profit of the cyber-insurer is controlled by the value of $\gamma$. For a given weak level of the social externality, e.g., $\alpha=6.5\times10^{-4}$, there is a corresponding strategy where the cyber-insurer's profit is maximized. The strategy is marked by the arrowhead of ``Best response'' and constitutes the NE strategy of the cyber-insurer under the given social externality setting. We observe from Fig.~\ref{fig:best_response}(a) that as a function of the premium coefficient $\gamma$, the profit is unimodal, and thus the optimal solution can be obtained analytically. Similarly, there exists the optimal point in each of Figs.~\ref{fig:best_response}(b), (c), and (d) for the blockchain provider. The optimal points are marked by the arrowhead of ``Best response'' and constitute the NE strategy of  the blockchain provider under different levels of social externality. Note that all the results shown in the figures are obtained given that the users play the NE in the user-level game.

\begin{figure}[!]
\centering
\includegraphics[width=0.4\textwidth,trim=5 5 5 5,clip]{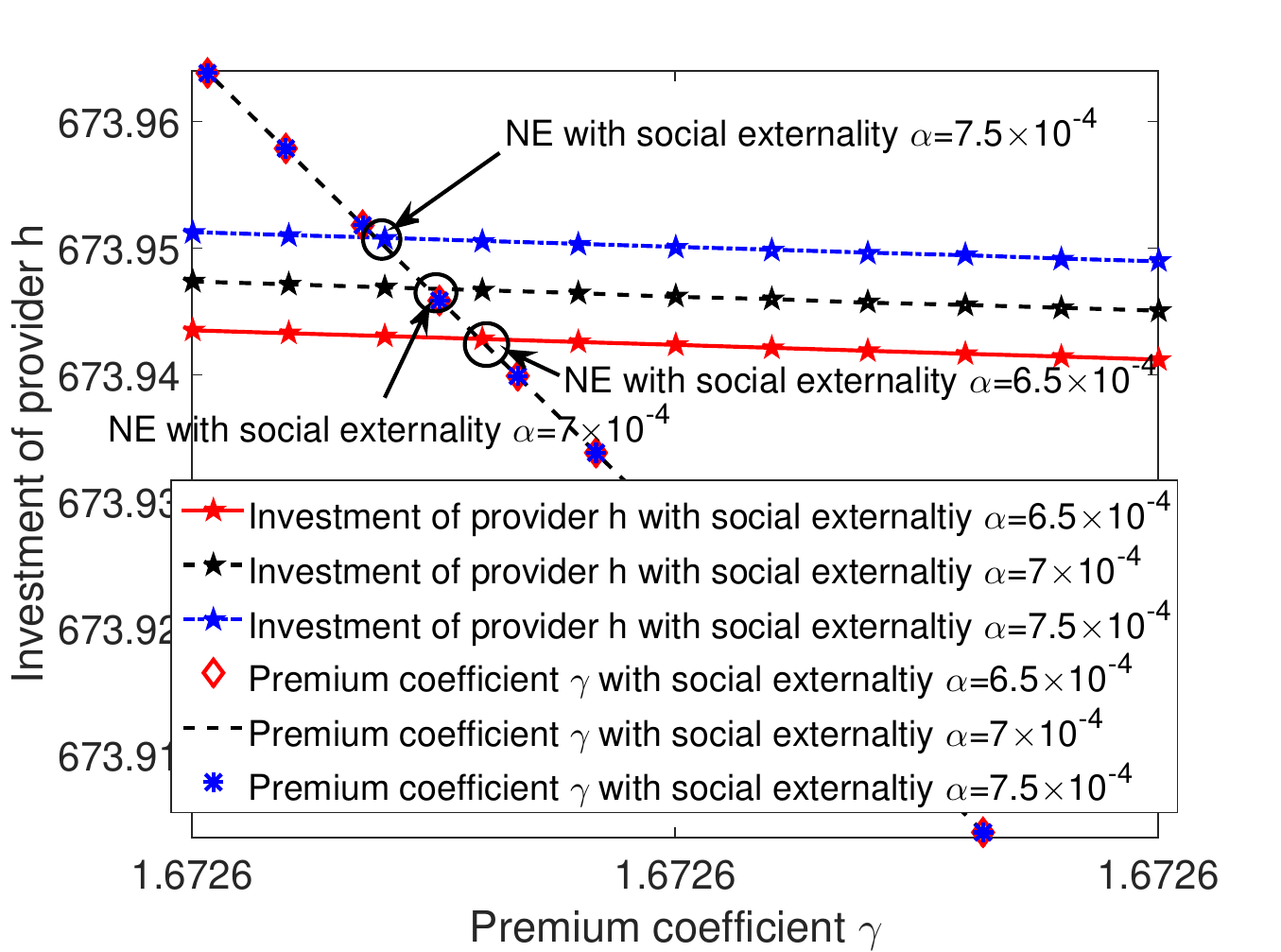}
\caption{Nash equilibrium}
\label{fig:ne}
\end{figure}

Figure~\ref{fig:ne} illustrates the equilibrium strategies for the investment of the blockchain provider and the premium coefficient $\gamma$ for the cyber-insurer under different levels of social externality. The NE is the point at which the best responses for the blockchain provider and cyber-insurer intersect. Again, this is given that the users play the NE in the user-level game. Under different levels of social externality, different NE are observed. As the level of social externality grows higher, i.e., a decision of one user has stronger effect to the decisions of other users, the other users are more likely to buy the same service if one user buys the service. As expected, when the level of social externality grows higher, the users are more likely to also buy the service, the blockchain provider has more money and more incentives to invest in the infrastructure and accordingly the premium coefficient $\gamma$ decreases (see Fig.~\ref{fig:ne}). The reason for this result is explained in the subsequent discussions.

\subsubsection{The impact of the number of users}
\label{subsubsec:impact_of_number_of_user}

\begin{figure}[!]
 \centering
 \includegraphics[width=0.5\textwidth,trim=100 15 115 20,clip]{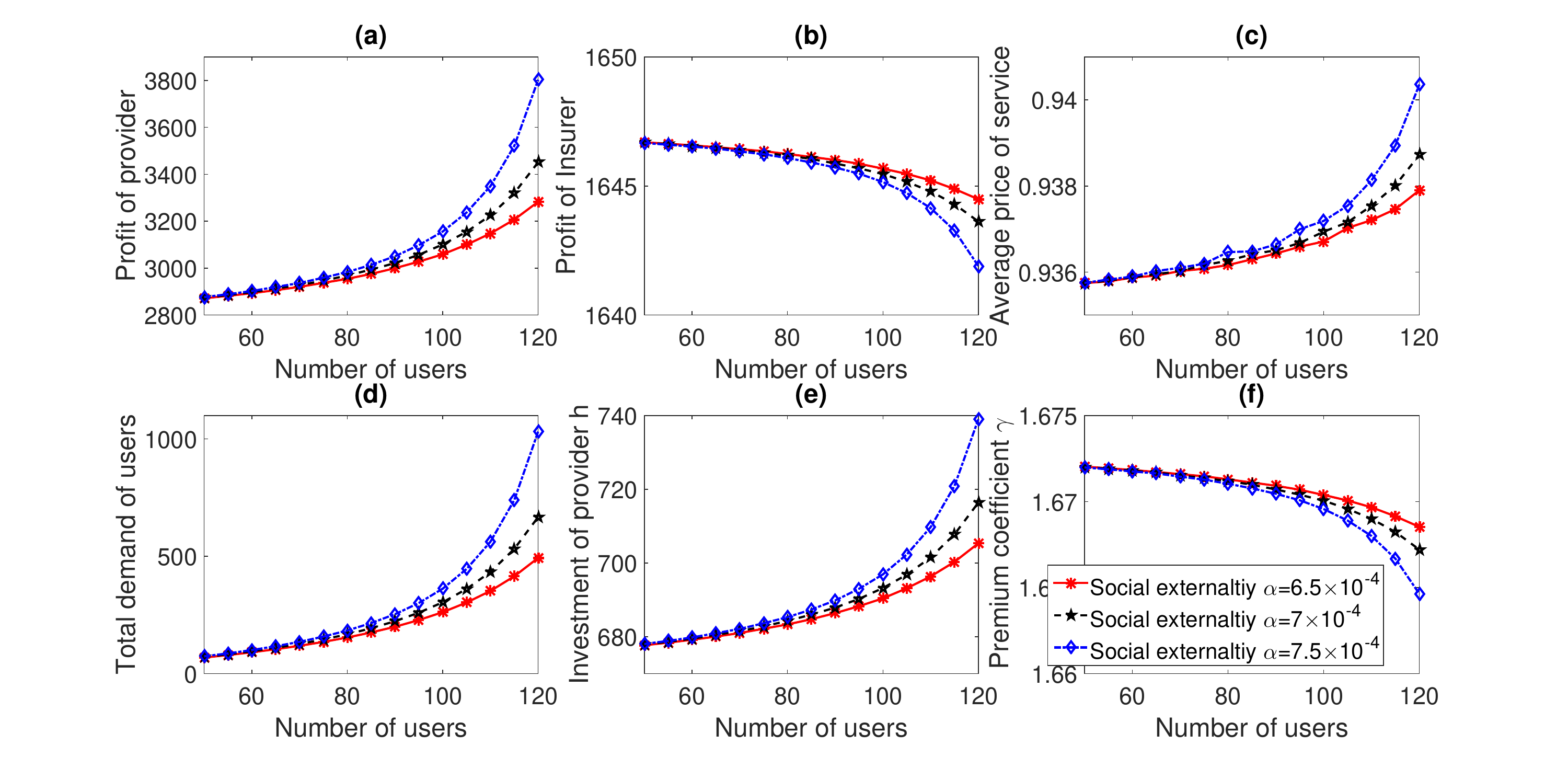}
 \caption{The results with increasing number of users.}
 \label{fig:Number_of_Users}
\end{figure}

We first evaluate the impacts by the number of users on the payoff of the market entities in Fig.~\ref{fig:Number_of_Users}, where the number of users increases from $50$ to $120$. Then, we evaluate the performance under three levels of social externality, e.g., $6.5 \times 10^{-4}$, $7 \times 10^{-4}$, and $7.5 \times 10^{-4}$ which represent weak, medium, and strong levels of social externality, respectively. As expected, the profit of the blockchain provider increases significantly when the social externality becomes stronger. As the number of users increases, the profit of the blockchain provider also increases under the given social externality settings. Moreover, the increase of the profit of the blockchain provider becomes larger when the social externality is stronger. Intuitively, the reason is that the social externality stimulates the demand of each user. In return, the blockchain provider can raise the price of service accordingly and hence improves its profit. Meanwhile, as the social externality becomes stronger, the blockchain provider also achieves greater profit. This indicates that the attack may incur more loss to the blockchain provider. Therefore, the blockchain provider has a higher incentive to invest in the infrastructure to prevent double-spending attack. This explains the result that the investment by the blockchain provider increases at a higher rate with $\alpha=7.5 \times 10^{-4}$ than with $\alpha=6.5 \times 10^{-4}$. As a result, the stronger social externality reduces the successful attack probability at a higher rate. Correspondingly, as shown in Fig.~\ref{fig:Number_of_Users}(f), the premium coefficient $\gamma$ decreases at a higher rate with the stronger social externality. This result indicates that the premium also decreases at a higher rate with stronger social externality. Thus, the cyber-insurer's profit decreases at a higher rate when the social externality becomes stronger as shown in Fig.~\ref{fig:Number_of_Users}(b).

\subsubsection{The impact of social externality}
\label{subsubsec:impact_of_social_externality}

\begin{figure}[!]
 \centering
 \includegraphics[width=0.5\textwidth,trim=100 15 120 20,clip]{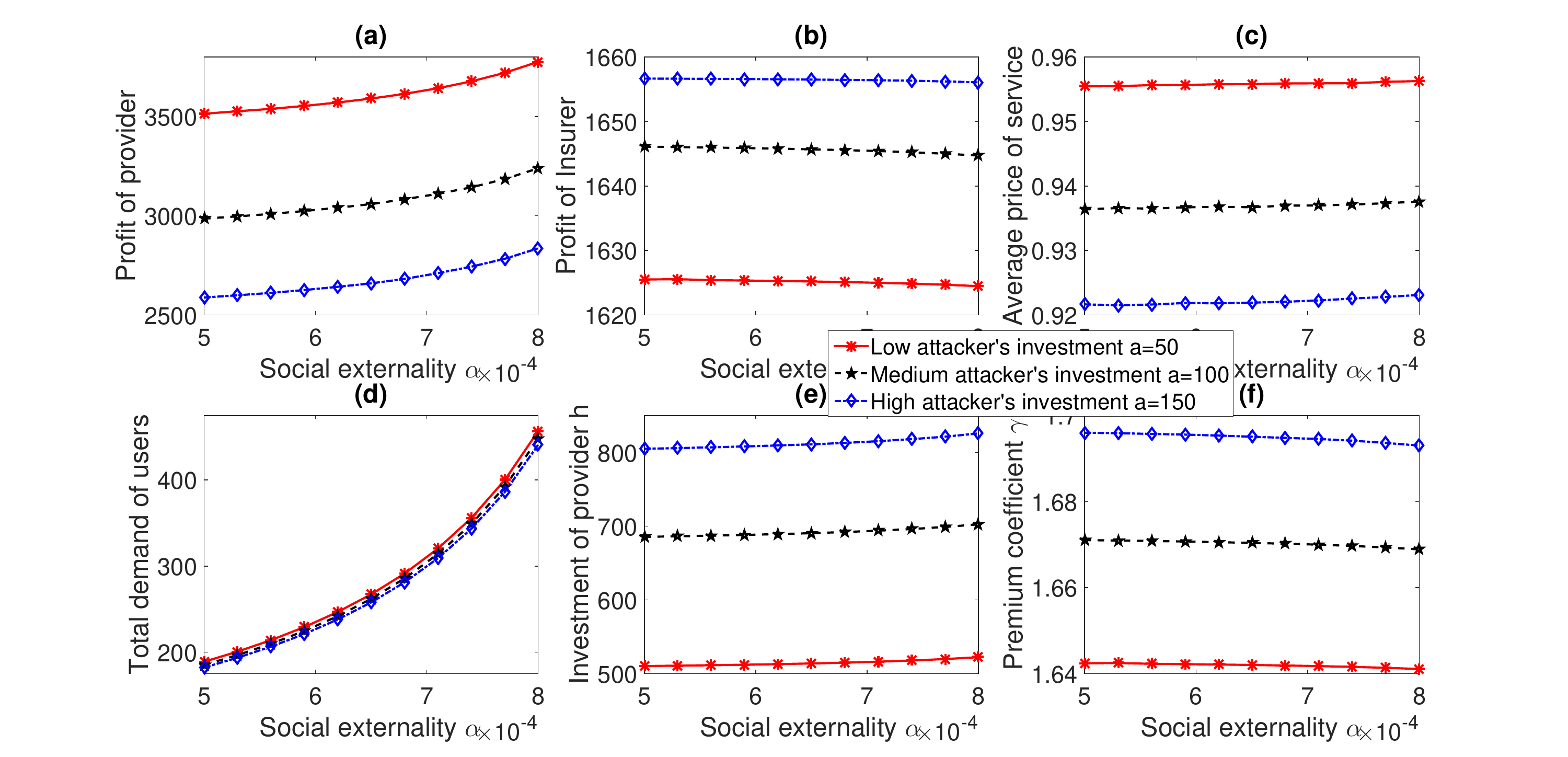}
 \caption{The results with increasing social externality.}
 \label{fig:Social_Externality}
\end{figure}

Figure~\ref{fig:Social_Externality} illustrates the impact of social externality on the payoffs of the three market entities with $100$ users. As expected, the users' total service demand increases as the social externality becomes stronger (see Fig.~\ref{fig:Social_Externality}(d)). From Figs.~\ref{fig:Social_Externality}(a) and~(b), we observe that the profit of the blockchain provider increases while the profit of the cyber-insurer decreases. The reason is that the users with stronger social externality are more sensitive to the security level of the blockchain service, and the security level depends on the investment ratio of the blockchain provider. Therefore, the blockchain provider will raise its investment as the social externality increases, which decreases the probability of successful double-spending attacks. Accordingly, the cyber-insurer reduces its premium as shown in Fig.~\ref{fig:Social_Externality}(f), and this will lead to the decrease of the cyber-insurer's profit. Moreover, since the investment by the blockchain provider increases, the security level of the blockchain service is improved, and thus the service demand of users increases. This situation results in the increase of the blockchain provider's profit.

Additionally, recall that the investment ratio of the blockchain provider depends on not only the investment from the blockchain provider, but also the computing resource owned by the attackers. For this reason, in our performance evaluation, we also vary the attacker's investment in computing resource by considering three cases with $a=50$, $a=100$ and $a=150$, respectively. From Fig.~\ref{fig:Social_Externality}(f), we observe that the decreasing rate of the curves becomes higher as the attacker's investment increases, which means that decreasing rate of the premium becomes higher as the attacker's investment increases. The reason is that the increasing rate of the investment by the blockchain provider becomes faster as the attacker's investment increases, and this also results in a higher decreasing rate of the successful attack probability. Consequently, the decreasing rate of the probability of paying claim is the highest with $a=150$ which is the largest attacker's investment.

\subsubsection{The impact of attacker's computing resource}
\label{subsubsec:impact_of_attacker_investment}

\begin{figure}[!]
 \centering
 \includegraphics[width=0.5\textwidth,trim=90 15 115 20,clip]{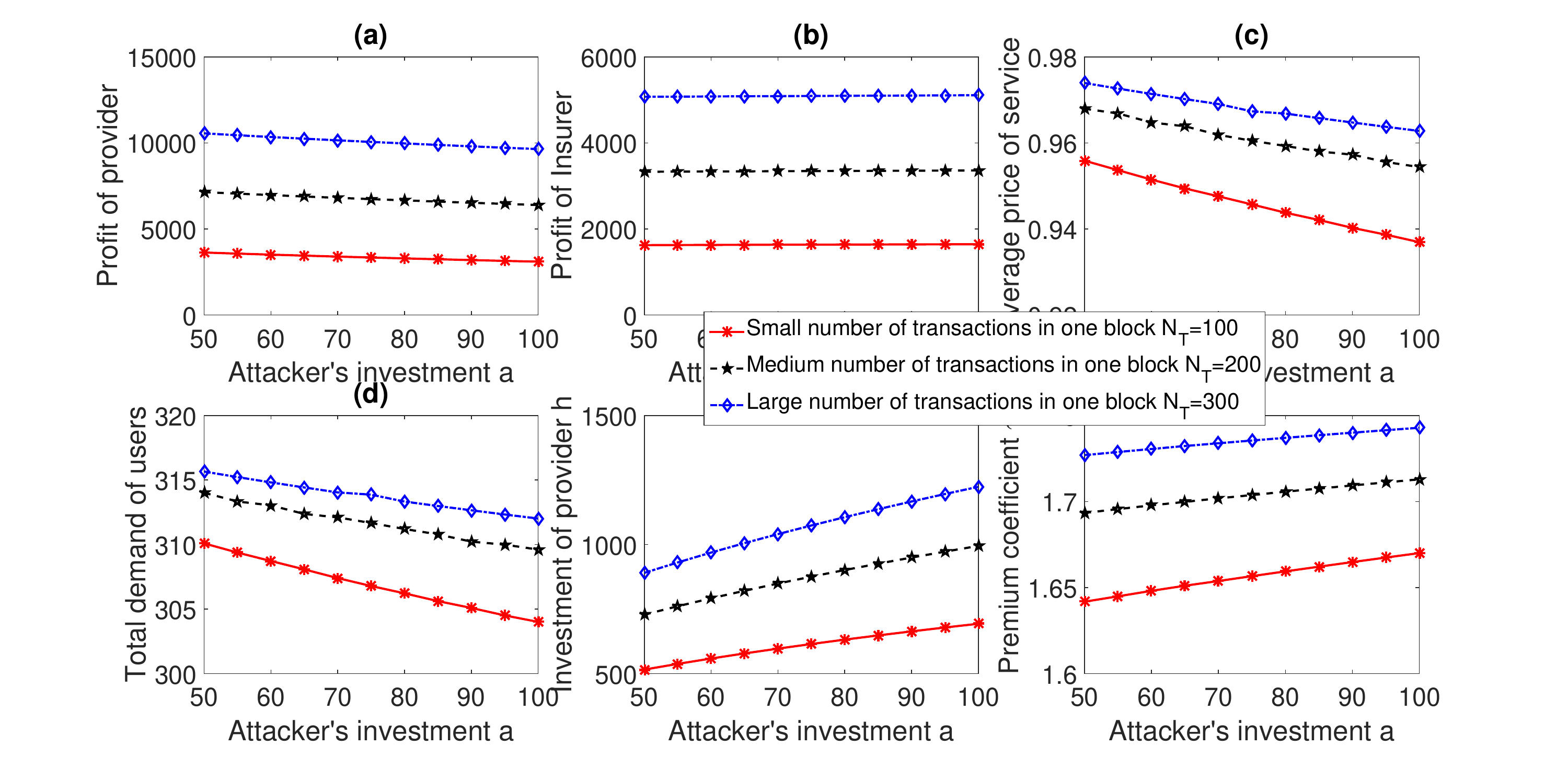}
 \caption{The result with increasing attacker's computing resource.}
 \label{fig:Increasing_attacker_investment}
\end{figure}

Finally, we evaluate in Fig.~\ref{fig:Increasing_attacker_investment} the impact of the attackers' computing resource on the performance of the users, the blockchain provider and the cyber-insurer. We consider three different situations with different sizes of a block, i.e., the number of transactions included in that block, with $N_{\rm{T}}=100$, $N_{\rm{T}}=200$ and $N_{\rm{T}}=300$. We observe from Figs.~\ref{fig:Increasing_attacker_investment}(a) and~(b) that as the attacker's computing resource increases, the profit of the blockchain provider decreases and the profit of the cyber-insurer remains unchanged. The reason is that the blockchain provider needs to increase its infrastructure investment when the attacker's computing resource increases (see Fig.~\ref{fig:Increasing_attacker_investment}(e)). Otherwise, the successful attack probability will significantly increase, and it results in the increase in the cost of the blockchain provider. With an increasing computing resource controlled by the attackers, the investment ratio of the blockchain provider cannot remain as high as before. This result is illustrated in Fig.~\ref{fig:Increasing_attacker_investment}(e) with ${\left. {\frac{{{h^*}}}{{a + {h^*}}}} \right|_{a = 50}} \approx \frac{18}{19} > {\left. {\frac{{{h^*}}}{{a + {h^*}}}} \right|_{a = 100}} \approx \frac{12}{13}$ when $N_{\rm{T}} = 300$. Therefore, as the attackers' computing resource increases, the successful attack probability and consequently the probability of the cyber-insurer paying the claim increases accordingly. Then, as shown in Fig.~\ref{fig:Increasing_attacker_investment}(f), the cyber-insurer increases the premium to keep its profit unchanged. Ultimately, the cost of the blockchain provider also increases. Moreover, as the attackers' computing resource increases, even when the blockchain provider reduces the price of its service to attract more users, the total demand from the users still decreases (see Figs.~\ref{fig:Increasing_attacker_investment}(c) and~(d)). Consequently, this reduces the revenue and profit of the blockchain provider.

Furthermore, we show that the impact of the attackers' computing resource on the other parties in the market is subject to the number of transactions in one block $N_{\rm{T}}$. With the fixed transaction fee and compensation rate, the more transactions in a single block is, the more reward that the blockchain provider obtains from mining a block. Moreover, since the compensation price of one block increases as the number of transactions in one block increases, the more compensation it will pay to the users once the attacks happen. Then, the blockchain provider has more incentive to invest in the computing resource. This is consistent with Fig.~\ref{fig:Increasing_attacker_investment}(e), where the increasing rate of the investment by the blockchain provider is higher with $N_{\rm{T}}=300$ than that with $N_{\rm{T}}=100$. On the other hand, as the attacker's computing resource increases, the increasing rate of the successful attack probability is lower with $N_{\rm{T}}=300$ than that with $N_{\rm{T}}=100$. Consequently, the increasing rate of the premium is lower with $N_{\rm{T}}=300$ than that with $N_{\rm{T}}=100$ (see Fig.~\ref{fig:Increasing_attacker_investment}(f)). Moreover, as $N_{\rm{T}}$ increases,
the increasing rate of the successful attack probability will decrease, and the users' demand will be less affected. This is consistent with Figs.~\ref{fig:Increasing_attacker_investment}(c) and~(d), where the decreasing rates of both the total user demand and the service price shrink as $N_{\rm{T}}$ increases.

\section{Conclusion}
\label{sec:conclusion}
In this paper, we have proposed a risk management framework of the blockchain service market by introducing the cyber-insurance as a mean for protecting financially the blockchain provider from double-spending attacks. We have modeled the interaction among the blockchain provider, the cyber-insurer and the users in the market as a two-stage Stackelberg game. For the blockchain provider, we have considered the problem of balancing between the proactive protection strategy, i.e., investing in computing power, and the reactive protection strategy, i.e., purchasing the cyber-insurance. For the users, we have considered the impact of both the social externality and the service security on the users' valuation of the blockchain service. In particular, for the cyber-insurer, we have incorporated the risk-adjusted pricing mechanism for premium adaptation. We have studied the equilibrium strategies of the three parties in the market using backward induction. We have analytically examined the conditions for the Stackelberg game to exist and to be unique. Furthermore, we have conducted extensive simulations to evaluate the performance of the market entities at the equilibrium. For the future work, we will investigate the long-run competition between the blockchain provider and cyber-insurer.

\ifCLASSOPTIONcaptionsoff
  \newpage
\fi

\bibliography{bibfile}

\newpage

\appendices
\section{Proof of Theorem~\ref{Th:existence_for_leader_game}}
\label{subsec:existence_leader_game}

\begin{proof}
Substituting the second partial derivatives of \begin{footnotesize}$\Pi_{\rm{P}}$\end{footnotesize} in~(\ref{Eq:second_partial_drivative_of_profit}), the Hessian matrix of the blockchain provider's payoff is
\begin{footnotesize}\begin{equation}\label{Eq:hessian_matrix_provider}
{{\bf{H}}_{\rm{P}}} = \left[ {\begin{aligned}
{\frac{{{\partial ^2}{\Pi_{\rm{P}}}}}{{\partial {{\bf{p}}^2}}}}&{\quad\frac{{{\partial ^2}{\Pi_{\rm{P}}}}}{{\partial {\bf{p}}\partial \bar h}}}\\
{\frac{{{\partial ^2}{\Pi_{\rm{P}}}}}{{\partial \bar h\partial {\bf{p}}}}}&{\quad\frac{{{\partial ^2}{\Pi_{\rm{P}}}}}{{\partial {{\bar h}^2}}}}
\end{aligned}} \right] = \left[ {\begin{aligned}
{ - {{\left( {{\bf{I}} - \alpha {\bf{G}}} \right)}^{ - 1}}}&{\quad{{\left( {{\bf{I}} - \alpha {\bf{G}}} \right)}^{ - 1}}{\bf{1}}}\\
{{{\bf{1}}^ \top }{{\left( {{\bf{I}} - \alpha {\bf{G}}} \right)}^{ - 1}}}&{\quad - \frac{{2a}}{{{{\left( {1 - \bar h} \right)}^3}}}}
\end{aligned}} \right].
\end{equation}\end{footnotesize}
We can partition the Hessian matrix in (\ref{Eq:hessian_matrix_provider}) into two matrices as follows:
\begin{equation}\label{Eq:partition_hessian_matrix_provider}
\footnotesize{\begin{aligned}
{{\bf{H}}_{\rm{P}}} &= \underbrace {\left[ {\begin{aligned}
{ - {{\left( {{\bf{I}} - \alpha {\bf{G}}} \right)}^{ - 1}}}&{\quad{{\left( {{\bf{I}} - \alpha {\bf{G}}} \right)}^{ - 1}}{\bf{1}}}\\
{{{\bf{1}}^ \top }{{\left( {{\bf{I}} - \alpha {\bf{G}}} \right)}^{ - 1}}}&{\quad - \frac{a}{{{{\left( {1 - \bar h} \right)}^3}}}}
\end{aligned}} \right]}_{{{\bf{H}}_{{\rm{P}}_1}}} + \underbrace {\left[ {\begin{array}{*{20}{c}}
{\bf{O}}&{\bf{0}}\\
{{{\bf{0}}^ \top }}&{ - \frac{a}{{{{\left( {1 - \bar h} \right)}^3}}}}
\end{array}} \right]}_{{{\bf{H}}_{{\rm{P}}_2}}}\\
&= {{\bf{H}}_{{\rm{P}}_1}} + {{\bf{H}}_{{\rm{P}}_2}},
\end{aligned}}
\end{equation}
where \begin{footnotesize}${{\bf{H}}_{{\rm{P}}_2}}$\end{footnotesize} is obviously a negative semidefinite matrix.

Let ${{\bf{Q}}_{\rm{P}}} = \left[ {\begin{array}{*{20}{c}}
{\bf{I}}&{\bf{1}}\\
{{{\bf{0}}^ \top }}&1
\end{array}} \right]$. The congruent matrix of ${{\bf{H}}_{{\rm{P}}_1}}$ is:
\begin{footnotesize}\begin{equation}\label{Eq:congruent_matrix_of_H_P_1}
{{\bf{Q}}_{\rm{P}}}^\top {{\bf{H}}_{{\rm{P}}_1}} {{\bf{Q}}_{\rm{P}}}=\left[ {\begin{array}{*{20}{c}}
{ - {{\left( {{\bf{I}} - \alpha {\bf{G}}} \right)}^{ - 1}}}&{\bf{0}}\\
{{{\bf{0}}^ \top }}&{ - \frac{a}{{{{\left( {1 - \bar h} \right)}^3}}} + {{\bf{1}}^ \top }{{\left( {{\bf{I}} - \alpha {\bf{G}}} \right)}^{ - 1}}{\bf{1}}}
\end{array}} \right].
\end{equation}\end{footnotesize}
Therefore, the negative definiteness of \begin{footnotesize}${{\bf{H}}_{{\rm{P}}_1}}$\end{footnotesize} depends on the negative definiteness of~(\ref{Eq:congruent_matrix_of_H_P_1}). According to the proof of Theorem~\ref{Th:existent_and_uniqueness_for_user_game}, all the eigenvalues of \begin{footnotesize}${{{\left( {{\bf{I}} - \alpha {\bf{G}}} \right)}^{ - 1}}}$\end{footnotesize} belong to \begin{footnotesize}$\left(0, 1\right)$\end{footnotesize}. This means that \begin{footnotesize}${ - {{\left( {{\bf{I}} - \alpha {\bf{G}}} \right)}^{ - 1}}}$\end{footnotesize} in (\ref{Eq:congruent_matrix_of_H_P_1}) is a negative definite matrix.
We know that the domain of \begin{footnotesize}${\bar{h}}$\end{footnotesize} is \begin{footnotesize}$\left[\frac{1}{2}, 1\right)$\end{footnotesize}, and the inequality equation in~(\ref{eq_new_intermediate}) will be satisfied.
\begin{equation}\label{eq_new_intermediate}
\footnotesize{\begin{aligned}
&- \frac{a}{{{{\left( {1 - \bar h} \right)}^3}}} + {{\bf{1}}^ \top }{\left( {{\bf{I}} - \alpha {\bf{G}}} \right)^{ - 1}}{\bf{1}} <  - \frac{a}{{{{\left( {1 - \frac{1}{2}} \right)}^3}}} + {{\bf{1}}^ \top }{\left( {{\bf{I}} - \alpha {\bf{G}}} \right)^{ - 1}}{\bf{1}} \\
=&  - 8a + {{\bf{1}}^ \top }{\left( {{\bf{I}} - \alpha {\bf{G}}} \right)^{ - 1}}{\bf{1}}.
\end{aligned}}
\end{equation}
\begin{footnotesize}$ - 8a + {{\bf{1}}^ \top }{\left( {{\bf{I}} - \alpha {\bf{G}}} \right)^{ - 1}}{\bf{1}} < 0 $\end{footnotesize} is satisfied according to Theorem~\ref{Th:existence_for_leader_game} and hence the second block in~(\ref{Eq:congruent_matrix_of_H_P_1}), i.e., \begin{footnotesize}$ - \frac{a}{{{{\left( {1 - \bar h} \right)}^3}}} + {{\bf{1}}^ \top }{\left( {{\bf{I}} - \alpha {\bf{G}}} \right)^{ - 1}}{\bf{1}}$\end{footnotesize}, is also negative. With the two negative definite blocks on the diagonal positions and zero blocks on the off-diagonal positions, the matrix in~(\ref{Eq:congruent_matrix_of_H_P_1}) is a negative definite matrix. Therefore, \begin{footnotesize}${{\bf{H}}_{{\rm{P}}_1}}$\end{footnotesize} is a negative definite matrix. Then, as the sum of one negative definite matrix and one negative semidefinite matrix, i.e., \begin{footnotesize}${{\bf{H}}_{{\rm{P}}_1}}$\end{footnotesize} and \begin{footnotesize}${{\bf{H}}_{{\rm{P}}_2}}$\end{footnotesize}, respectively, the Hessian matrix \begin{footnotesize}${{\bf{H}}_{\rm{P}}}$\end{footnotesize} in~(\ref{Eq:hessian_matrix_provider}) is a negative definite matrix. Therefore, \begin{footnotesize}$\Pi_{\rm{P}}$\end{footnotesize} is a strictly concave function of \begin{footnotesize}$\left[{\bf{p}}^\top, {\bar h}\right]^\top$\end{footnotesize}.

In what follows, we prove that $\Pi_{\rm{I}}$ is also a concave function of its own decision variable $\gamma$. Recall the second partial derivative of $\Pi_{\rm{I}}$ on $\gamma$ in~(\ref{Eq:second_partial_drivative_of_profit}) as follows:
\begin{footnotesize}\begin{equation}
\begin{aligned}
\frac{{{\partial ^2}{\Pi _{\rm{I}}}}}{{\partial {\gamma ^2}}} =& \frac{2}{{{\gamma ^3}}}\frac{T}{{{T_{\rm{0}}}}}{N_{\rm{T}}}q\int\nolimits_{1/2}^1 {{{\left[ {1 - \int\nolimits_{1/2}^t {{\rm{P}}\left( \theta  \right){\rm{d}}\theta } } \right]}^{1/\gamma }}\ln \left[ {1 - \int\nolimits_{1/2}^t {{\rm{P}}\left( \theta  \right){\rm{d}}\theta } } \right]{\rm{d}}t} \\
 +& \frac{1}{{{\gamma ^4}}}\frac{T}{{{T_{\rm{0}}}}}{N_{\rm{T}}}q\int\nolimits_{1/2}^1 {{{\left[ {1 - \int\nolimits_{1/2}^t {{\rm{P}}\left( \theta  \right){\rm{d}}\theta } } \right]}^{1/\gamma }}{{\ln }^2}\left[ {1 - \int\nolimits_{1/2}^t {{\rm{P}}\left( \theta  \right){\rm{d}}\theta } } \right]{\rm{d}}t} \\
 -& {\left( {\overline h  - \frac{1}{2}} \right)^3}\left[ {\beta \left( {\beta  + 1} \right){\gamma ^{\beta  - 1}} - \beta \left( {\beta  - 1} \right){\gamma ^{\beta  - 2}}} \right]\\
 =& \frac{1}{{{\gamma ^3}}}\frac{T}{{{T_{\rm{0}}}}}{N_{\rm{T}}}q\int\nolimits_{1/2}^1 {\underbrace {{{\left[ {1 - \int\nolimits_{1/2}^t {{\rm{P}}\left( \theta  \right){\rm{d}}\theta } } \right]}^{1/\gamma }}}_{ > 0}\underbrace {\ln \left[ {1 - \int\nolimits_{1/2}^t {{\rm{P}}\left( \theta  \right){\rm{d}}\theta } } \right]}_{ < 0}{\rm{d}}t}  \\
 +& \frac{1}{{{\gamma ^3}}}\frac{T}{{{T_{\rm{0}}}}}{N_{\rm{T}}}q \int\nolimits_{1/2}^1 {\underbrace {{{\left[ {1 - \int\nolimits_{1/2}^t {{\rm{P}}\left( \theta  \right){\rm{d}}\theta } } \right]}^{1/\gamma }}}_{ > 0}\underbrace {\ln \left[ {1 - \int\nolimits_{1/2}^t {{\rm{P}}\left( \theta  \right){\rm{d}}\theta } } \right]}_{ < 0}}\\
 \times&{\left\{ {1 + \frac{1}{\gamma }\ln \left[ {1 - \int\nolimits_{1/2}^t {{\rm{P}}\left( \theta  \right){\rm{d}}\theta } } \right]} \right\}{\rm{d}}t} \\
 -& \underbrace {{{\left( {\overline h  - \frac{1}{2}} \right)}^3}}_{ > 0}\underbrace {\left[ {\left( {\beta  + 1} \right)\gamma  - \left( {\beta  - 1} \right)} \right]}_{ > 0}\beta {\gamma ^{\beta  - 2}}.
\end{aligned}
\end{equation}\end{footnotesize}
To ensure the concavity of ${\Pi _{\rm{I}}}$ with respect to \begin{footnotesize}$\gamma$\end{footnotesize}, we need the following inequality to be satisfied:
\begin{footnotesize}\begin{equation}
  \label{eq_proof_intermediate2}
  {\left\{ {1 + \frac{1}{\gamma }{\ln}\left[ {1 - \int\nolimits_{1/2}^t {{\rm{P}}\left( \theta  \right){\rm{d}}\theta } } \right]} \right\}} \ge 0.
\end{equation}\end{footnotesize}
From (\ref{eq_proof_intermediate2}), we have
\begin{footnotesize}\begin{equation}
  \label{eq_intermediate_inequality}
\begin{aligned}
&\left\{ {1 + \frac{1}{\gamma }\ln \left[ {1 - \int_{1/2}^t {{\rm{P}}\left( \theta  \right){\rm{d}}\theta } } \right]} \right\} \ge 0\\
\Leftrightarrow& \frac{1}{\gamma }\ln \left[ {1 - \int_{1/2}^t {{\rm{P}}\left( \theta  \right){\rm{d}}\theta } } \right] \ge  - 1\\
\Leftrightarrow& \frac{1}{\gamma }\ln \left[ {1 - \int_{1/2}^t {{\rm{P}}\left( \theta  \right){\rm{d}}\theta } } \right] \\
&\ge \frac{1}{\gamma }\ln \left[ {\int_0^1 {{\rm{P}}\left( \theta  \right){\rm{d}}\theta }  - \int_{1/2}^1 {{\rm{P}}\left( \theta  \right){\rm{d}}\theta } } \right] \ge  - 1\\
\Rightarrow& \frac{1}{\gamma }\ln \left[ {\int_0^{1/2} {{\rm{P}}\left( \theta  \right){\rm{d}}\theta } } \right] \ge  - 1 \Leftrightarrow \frac{1}{\gamma }\ln \frac{1}{2} \ge  - 1 \Leftrightarrow \gamma  \ge \ln 2
\end{aligned}
\end{equation}\end{footnotesize}
From (\ref{eq_intermediate_inequality}) we learn that the sufficient condition to ensure the inequality in (\ref{eq_proof_intermediate2}) is \begin{footnotesize}$\gamma  \ge \ln 2$\end{footnotesize}. Such a condition holds on  \begin{footnotesize}$\left(1, \gamma^{\rm{u}}\right]$\end{footnotesize}. Therefore, \begin{footnotesize}$\Pi_{\rm{I}}$\end{footnotesize} is a concave function on \begin{footnotesize}$\gamma$\end{footnotesize}. Then, the profit functions of the blockchain provider and the cyber-insurer, i.e., \begin{footnotesize}$\Pi_{\rm{P}}$\end{footnotesize} and \begin{footnotesize}$\Pi_{\rm{I}}$\end{footnotesize}, are concave with respect to their own decision variables, i.e., \begin{footnotesize}$\left[{\bf{p}}^\top, {\bar h}\right]^\top$\end{footnotesize} and \begin{footnotesize}$\gamma$\end{footnotesize}, respectively. According to~\cite{debreu1952social}, the NE in the leader-level game \begin{footnotesize}${\cal{G}}_{\rm{L}}$\end{footnotesize} exists hence the Stackbelberg equilibrium exists. The proof is completed.
\end{proof}

\section{Proof of Theorem~\ref{Th:uniqueness_for_leader_game}}
\label{subsec:uniqueness_leader_game}

\begin{figure*}
\begin{footnotesize}
\begin{equation}\label{Eq:jacobian_matrix_leader_game}
\begin{aligned}
{\bf{J}}=&\left[ {\begin{aligned}
{\frac{{{\partial ^2}{\Pi _p}}}{{\partial {{\bf{p}}^2}}}}&{\quad\frac{{{\partial ^2}{\Pi _p}}}{{\partial {\bf{p}}\partial \bar h}}}&{\frac{{{\partial ^2}{\Pi _p}}}{{\partial {\bf{p}}\partial \gamma }}}\\
{\frac{{{\partial ^2}{\Pi _p}}}{{\partial \bar h\partial {\bf{p}}}}}&{\quad\frac{{{\partial ^2}{\Pi _p}}}{{\partial {{\bar h}^2}}}}&{\frac{{{\partial ^2}{\Pi _p}}}{{\partial \bar h\partial \gamma }}}\\
{\frac{{{\partial ^2}{\Pi _I}}}{{\partial \gamma \partial {\bf{p}}}}}&{\quad\frac{{{\partial ^2}{\Pi _I}}}{{\partial \gamma \partial \bar h}}}&{\frac{{{\partial ^2}{\Pi _I}}}{{\partial {\gamma ^2}}}}
\end{aligned}} \right] + {\left[ {\begin{aligned}
{\frac{{{\partial ^2}{\Pi _p}}}{{\partial {{\bf{p}}^2}}}}&{\quad\frac{{{\partial ^2}{\Pi _p}}}{{\partial {\bf{p}}\partial \bar h}}}&{\frac{{{\partial ^2}{\Pi _p}}}{{\partial {\bf{p}}\partial \gamma }}}\\
{\frac{{{\partial ^2}{\Pi _p}}}{{\partial \bar h\partial {\bf{p}}}}}&{\quad\frac{{{\partial ^2}{\Pi _p}}}{{\partial {{\bar h}^2}}}}&{\frac{{{\partial ^2}{\Pi _p}}}{{\partial \bar h\partial \gamma }}}\\
{\frac{{{\partial ^2}{\Pi _I}}}{{\partial \gamma \partial {\bf{p}}}}}&{\quad\frac{{{\partial ^2}{\Pi _I}}}{{\partial \gamma \partial \bar h}}}&{\frac{{{\partial ^2}{\Pi _I}}}{{\partial {\gamma ^2}}}}
\end{aligned}} \right]^\top}
 = \left[ {\begin{array}{*{20}{c}}
{2\frac{{{\partial ^2}{\Pi _p}}}{{\partial {{\bf{p}}^2}}}}&{\frac{{{\partial ^2}{\Pi _p}}}{{\partial {\bf{p}}\partial \bar h}} + {{\left( {\frac{{{\partial ^2}{\Pi _p}}}{{\partial \bar h\partial {\bf{p}}}}} \right)}^\top}}&{\frac{{{\partial ^2}{\Pi _p}}}{{\partial {\bf{p}}\partial \gamma }} + {{\left( {\frac{{{\partial ^2}{\Pi _I}}}{{\partial \gamma \partial {\bf{p}}}}} \right)}^\top}}\\
{\frac{{{\partial ^2}{\Pi _p}}}{{\partial \bar h\partial {\bf{p}}}} + {{\left( {\frac{{{\partial ^2}{\Pi _p}}}{{\partial {\bf{p}}\partial \bar h}}} \right)}^\top}}&{2\frac{{{\partial ^2}{\Pi _p}}}{{\partial {{\bar h}^2}}}}&{\frac{{{\partial ^2}{\Pi _p}}}{{\partial \bar h\partial \gamma }} + \frac{{{\partial ^2}{\Pi _I}}}{{\partial \gamma \partial \bar h}}}\\
{\frac{{{\partial ^2}{\Pi _I}}}{{\partial \gamma \partial {\bf{p}}}} + {{\left( {\frac{{{\partial ^2}{\Pi _p}}}{{\partial {\bf{p}}\partial \bar h}}} \right)}^\top}}&{\frac{{{\partial ^2}{\Pi _I}}}{{\partial \gamma \partial \bar h}} + \frac{{{\partial ^2}{\Pi _p}}}{{\partial \bar h\partial \gamma }}}&{2\frac{{{\partial ^2}{\Pi _I}}}{{\partial {\gamma ^2}}}}
\end{array}} \right]
\end{aligned}
\end{equation}
\end{footnotesize}
\begin{footnotesize}
\begin{equation}\label{Eq:specifical_jacobian_matrix_leader_game}
 =  \left[ {\begin{array}{*{20}{c}}
{ - 2{{\left( {{\bf{I}} - \alpha {\bf{G}}} \right)}^{ - 1}}}&{2{{\left( {{\bf{I}} - \alpha {\bf{G}}} \right)}^{ - 1}}{\bf{1}}}&{\bf{0}}\\
{2{{\bf{1}}^ \top }{{\left( {{\bf{I}} - \alpha {\bf{G}}} \right)}^{ - 1}}}&{ - \frac{{4a}}{{{{\left( {1 - \bar h} \right)}^3}}}}&\begin{aligned}
 &- 3{\left( {\overline h  - \frac{1}{2}} \right)^2}\\
 \times& \left[ {\left( {\beta  + 1} \right){\gamma ^\beta } - \beta {\gamma ^{\beta  - 1}}} \right]
\end{aligned}\\
{{{\bf{0}}^ \top }}&\begin{aligned}
 &- 3{\left( {\overline h  - \frac{1}{2}} \right)^2}\\
 \times &\left[ {\left( {\beta  + 1} \right){\gamma ^\beta } - \beta {\gamma ^{\beta  - 1}}} \right]
\end{aligned}&\begin{aligned}
&\frac{4}{{{\gamma ^3}}}\frac{T}{{{T_{\rm{0}}}}}{N_{\rm{T}}}q\int\nolimits_{1/2}^1 {{{\left[ {1 - \int\nolimits_{1/2}^t {{\rm{P}}\left( \theta  \right){\rm{d}}\theta } } \right]}^{1/\gamma }}\ln \left[ {1 - \int\nolimits_{1/2}^t {{\rm{P}}\left( \theta  \right){\rm{d}}\theta } } \right]{\rm{d}}t} \\
 +& \frac{2}{{{\gamma ^4}}}\frac{T}{{{T_{\rm{0}}}}}{N_{\rm{T}}}q\int\nolimits_{1/2}^1 {{{\left[ {1 - \int\nolimits_{1/2}^t {{\rm{P}}\left( \theta  \right){\rm{d}}\theta } } \right]}^{1/\gamma }}\ln \left[ {1 - \int\nolimits_{1/2}^t {{\rm{P}}\left( \theta  \right){\rm{d}}\theta } } \right]{\rm{d}}t} \\
 -& 2{\left( {\overline h  - \frac{1}{2}} \right)^3}\left[ {\beta \left( {\beta  + 1} \right){\gamma ^{\beta  - 1}} - \beta \left( {\beta  - 1} \right){\gamma ^{\beta  - 2}}} \right]
\end{aligned}
\end{array}} \right].
\end{equation}
\end{footnotesize}
\end{figure*}

\begin{proof}
According to Theorem 3.1 (Rosen's Theorem) in~\cite{peng2009summary}, the uniqueness of the NE in \begin{footnotesize}${\cal{G}}_{\rm{L}}$\end{footnotesize} holds if the Jacobian matrix constructed from the profit functions of the two players, i.e., \begin{footnotesize}$\mathbf{J}$\end{footnotesize} given in (\ref{Eq:jacobian_matrix_leader_game}), is negative definite. After substituting the relevant second partial derivatives given in~(\ref{Eq:second_partial_drivative_of_profit}) into (\ref{Eq:jacobian_matrix_leader_game}), we are able to rewrite (\ref{Eq:jacobian_matrix_leader_game}) as in (\ref{Eq:specifical_jacobian_matrix_leader_game}).

\begin{figure*}
\begin{footnotesize}\begin{equation}\label{Eq:partition_jacobian_matrix_leader_game}
\begin{aligned}
{\bf{J}} =&
\begin{small}\underbrace {\left[ {\begin{array}{*{20}{c}}
{ - 2{{\left( {{\bf{I}} - \alpha {\bf{G}}} \right)}^{ - 1}}}&{2{{\left( {{\bf{I}} - \alpha {\bf{G}}} \right)}^{ - 1}}{\bf{1}}}&{\bf{0}}\\
{2{{\bf{1}}^ \top }{{\left( {{\bf{I}} - \alpha {\bf{G}}} \right)}^{ - 1}}}&{ - \frac{{2a}}{{{{\left( {1 - \bar h} \right)}^3}}}}&0\\
{{{\bf{0}}^ \top }}&0&\begin{aligned}
&\frac{4}{{{\gamma ^3}}}\frac{T}{{{T_{\rm{0}}}}}{N_{\rm{T}}}q\int\nolimits_{1/2}^1 {{{\left[ {1 - \int\nolimits_{1/2}^t {{\rm{P}}\left( \theta  \right){\rm{d}}\theta } } \right]}^{1/\gamma }}\ln \left[ {1 - \int\nolimits_{1/2}^t {{\rm{P}}\left( \theta  \right){\rm{d}}\theta } } \right]{\rm{d}}t} \\
 +& \frac{2}{{{\gamma ^4}}}\frac{T}{{{T_{\rm{0}}}}}{N_{\rm{T}}}q\int\nolimits_{1/2}^1 {{{\left[ {1 - \int\nolimits_{1/2}^t {{\rm{P}}\left( \theta  \right){\rm{d}}\theta } } \right]}^{1/\gamma }}\ln \left[ {1 - \int\nolimits_{1/2}^t {{\rm{P}}\left( \theta  \right){\rm{d}}\theta } } \right]{\rm{d}}t}
\end{aligned}
\end{array}} \right]}_{{{\bf{J}}_1}}\end{small}\\
 &+ \underbrace {\left[ {\begin{array}{*{20}{c}}
{\bf{O}}&{\bf{0}}&{\bf{0}}\\
{{{\bf{0}}^ \top }}&{ - \frac{{2a}}{{{{\left( {1 - \bar h} \right)}^3}}}}&\begin{aligned}
 &- 3{\left( {\overline h  - \frac{1}{2}} \right)^2}\\
 \times& \left[ {\left( {\beta  + 1} \right){\gamma ^\beta } - \beta {\gamma ^{\beta  - 1}}} \right]
\end{aligned}\\
{{{\bf{0}}^ \top }}&\begin{aligned}
 &- 3{\left( {\overline h  - \frac{1}{2}} \right)^2}\\
 \times& \left[ {\left( {\beta  + 1} \right){\gamma ^\beta } - \beta {\gamma ^{\beta  - 1}}} \right]
\end{aligned}&{ - {{2\left( {\overline h  - \frac{1}{2}} \right)}^3}\left[ {\beta \left( {\beta  + 1} \right){\gamma ^{\beta  - 1}} - \beta \left( {\beta  - 1} \right){\gamma ^{\beta  - 2}}} \right]}
\end{array}} \right]}_{{{\bf{J}}_2}}.
\end{aligned}
\end{equation}\end{footnotesize}
\begin{footnotesize}\begin{equation}\label{Eq:congruent_matrix_J_2}
\begin{aligned}
{{{\bf{R}}_1}{{\bf{J}}_2}{{\bf{R}}_1}^ \top  = }
\left[ {\begin{array}{*{20}{c}}
{\bf{O}}&{\bf{0}}&{\bf{0}}\\
{{{\bf{0}}^ \top }}&\begin{aligned}
 &- \frac{{2a}}{{{{\left( {1 - \bar h} \right)}^3}}}\\
 +& \frac{{9{{\left( {\overline h  - \frac{1}{2}} \right)}^4}{{\left[ {\left( {\beta  + 1} \right){\gamma ^\beta } - \beta {\gamma ^{\beta  - 1}}} \right]}^2}}}{{{{2\left( {\overline h  - \frac{1}{2}} \right)}^3}\left[ {\beta \left( {\beta  + 1} \right){\gamma ^{\beta  - 1}} - \beta \left( {\beta  - 1} \right){\gamma ^{\beta  - 2}}} \right]}}
\end{aligned}&0\\
{{{\bf{0}}^ \top }}&0&{ - {{2\left( {\overline h  - \frac{1}{2}} \right)}^3}\left[ {\left( {\beta  + 1} \right)\gamma  - \left( {\beta  - 1} \right)} \right]\beta {\gamma ^{\beta  - 2}}}
\end{array}} \right].
\end{aligned}
\end{equation}\end{footnotesize}
\end{figure*}

Then, we partition the Jacobian matrix~(\ref{Eq:jacobian_matrix_leader_game}) into two matrices as shown in~(\ref{Eq:partition_jacobian_matrix_leader_game}). In \begin{footnotesize}${\bf{J}}_1$\end{footnotesize} given by (\ref{Eq:partition_jacobian_matrix_leader_game}), we have \begin{footnotesize}$\left[ {\begin{array}{*{20}{c}}
{ - 2{{\left( {{\bf{I}} - \alpha {\bf{G}}} \right)}^{ - 1}}}&{2{{\left( {{\bf{I}} - \alpha {\bf{G}}} \right)}^{ - 1}}{\bf{1}}}\\
{2{{\bf{1}}^ \top }{{\left( {{\bf{I}} - \alpha {\bf{G}}} \right)}^{ - 1}}}&{ - \frac{{2a}}{{{{\left( {1 - \bar h} \right)}^3}}}}
\end{array}} \right] = 2{{\bf{H}}_{{\rm{P}}_1}}$\end{footnotesize}, where \begin{footnotesize}${{\bf{H}}_{{\rm{P}}_1}}$\end{footnotesize} is the matrix in~(\ref{Eq:partition_hessian_matrix_provider}). As shown in the proof to Theorem~\ref{Th:existence_for_leader_game}, \begin{footnotesize}${{\bf{H}}_{{\rm{P}}_1}}$\end{footnotesize} is a negative definite matrix, and thus the first block in the matrix \begin{footnotesize}${\bf{J}}_1$\end{footnotesize} of~(\ref{Eq:partition_jacobian_matrix_leader_game}), i.e., \begin{footnotesize}$\left[ {\begin{array}{*{20}{c}}
{ - 2{{\left( {{\bf{I}} - \alpha {\bf{G}}} \right)}^{ - 1}}}&{2{{\left( {{\bf{I}} - \alpha {\bf{G}}} \right)}^{ - 1}}{\bf{1}}}\\
{2{{\bf{1}}^ \top }{{\left( {{\bf{I}} - \alpha {\bf{G}}} \right)}^{ - 1}}}&{ - \frac{{2a}}{{{{\left( {1 - \bar h} \right)}^3}}}}
\end{array}} \right]$\end{footnotesize}, is a negative definite matrix. For the second non-zero block in \begin{footnotesize}${\bf{J}}_1$\end{footnotesize}, the condition to ensure negative definiteness is
\begin{footnotesize}\begin{equation}
  \label{eq_intermediate_negative}
\begin{aligned}
0&>\frac{4}{{{\gamma ^3}}}\frac{T}{{{T_{\rm{0}}}}}{N_{\rm{T}}}q\int\nolimits_{1/2}^1 {{{\left[ {1 - \int\nolimits_{1/2}^t {{\rm{P}}\left( \theta  \right){\rm{d}}\theta } } \right]}^{1/\gamma }}{\rm{ln}}\left[ {1 - \int\nolimits_{1/2}^t {{\rm{P}}\left( \theta  \right){\rm{d}}\theta } } \right]{\rm{d}}t}  \\
&+ \frac{2}{{{\gamma ^4}}}\frac{T}{{{T_{\rm{0}}}}}{N_{\rm{T}}}q\int\nolimits_{1/2}^1 {{{\left[ {1 - \int\nolimits_{1/2}^t {{\rm{P}}\left( \theta  \right){\rm{d}}\theta } } \right]}^{1/\gamma }}{\rm{l}}{{\rm{n}}^2}\left[ {1 - \int\nolimits_{1/2}^t {{\rm{P}}\left( \theta  \right){\rm{d}}\theta } } \right]{\rm{d}}t} \\
 =& \frac{2}{{{\gamma ^3}}}\frac{T}{{{T_{\rm{0}}}}}{N_{\rm{T}}}q\int\nolimits_{1/2}^1 {{{\left[ {1 - \int\nolimits_{1/2}^t {{\rm{P}}\left( \theta  \right){\rm{d}}\theta } } \right]}^{1/\gamma }}{\rm{ln}}\left[ {1 - \int\nolimits_{1/2}^t {{\rm{P}}\left( \theta  \right){\rm{d}}\theta } } \right]{\rm{d}}t}\\
 &+\frac{2}{{{\gamma ^3}}}\frac{T}{{{T_{\rm{0}}}}}{N_{\rm{T}}}q\int\nolimits_{1/2}^1 {{\left[ {1 - \int\nolimits_{1/2}^t {{\rm{P}}\left( \theta  \right){\rm{d}}\theta } } \right]}^{1/\gamma }}\\
 &\times{\rm{ln}}\left[ {1 - \int\nolimits_{1/2}^t {{\rm{P}}\left( \theta  \right){\rm{d}}\theta } } \right]\left\{ {1 + \frac{1}{\gamma }\left[ {1 - \int\nolimits_{1/2}^t {{\rm{P}}\left( \theta  \right){\rm{d}}\theta } } \right]} \right\}{\rm{d}}t.
\end{aligned}
\end{equation}\end{footnotesize}
Following the same procedure in the proof to Theorem~\ref{Th:existence_for_leader_game}, from (\ref{eq_intermediate_negative}) we obtain the condition of negative definiteness as \begin{footnotesize}$\gamma \ge \ln2$\end{footnotesize}. Therefore, with the two negative definite blocks on the diagonal positions and all-zero blocks on the off-diagonal positions, \begin{footnotesize}${\bf{J}}_1$\end{footnotesize} in (\ref{Eq:partition_jacobian_matrix_leader_game}) is a negative definite matrix.

After defining
\begin{equation}
\footnotesize{\begin{aligned}
&{{\bf{R}}_1} =\\
&\left[ {\begin{array}{*{20}{c}}
{\bf{I}}&{\bf{0}}&{\bf{0}}\\
{{{\bf{0}}^ \top }}&1&\begin{aligned}
 &- 3{\left( {\overline h  - \frac{1}{2}} \right)^2}\left[ {\left( {\beta  + 1} \right){\gamma ^\beta } - \beta {\gamma ^{\beta  - 1}}} \right]\\
 \times &\left\{  - {{2\left( {\overline h  - \frac{1}{2}} \right)}^3}\right\}^{ - 1}\\
 \times&\left[ {\beta \left( {\beta  + 1} \right){\gamma ^{\beta  - 1}} - \beta \left( {\beta  - 1} \right){\gamma ^{\beta  - 2}}} \right]^{ - 1}
\end{aligned}\\
{{{\bf{0}}^ \top }}&0&1
\end{array}} \right],
\end{aligned}}
\end{equation}
we can obtain the congruent matrix of \begin{footnotesize}${\bf{J}}_2$\end{footnotesize} in~(\ref{Eq:partition_jacobian_matrix_leader_game}) as shown in~(\ref{Eq:congruent_matrix_J_2}). In (\ref{Eq:congruent_matrix_J_2}), we observe that \begin{footnotesize}${ - {{2\left( {\overline h  - \frac{1}{2}} \right)}^3}\left[ {\left( {\beta  + 1} \right)\gamma  - \left( {\beta  - 1} \right)} \right]\beta {\gamma ^{\beta  - 2}}}<0$\end{footnotesize}. Then, to ensure the negative semidefiniteness of \begin{footnotesize}${\bf{J}}_2$\end{footnotesize}, we need the following inequality to hold:
\begin{footnotesize}\begin{equation}
  \label{eq_intermediate_bound}
  - \frac{{2a}}{{{{\left( {1 - \bar h} \right)}^3}}} + \frac{{9{{\left( {\overline h  - \frac{1}{2}} \right)}^4}{{\left[ {\left( {\beta  + 1} \right){\gamma ^\beta } - \beta {\gamma ^{\beta  - 1}}} \right]}^2}}}{{{{2\left( {\overline h  - \frac{1}{2}} \right)}^3}\left[ {\beta \left( {\beta  + 1} \right){\gamma ^{\beta  - 1}} - \beta \left( {\beta  - 1} \right){\gamma ^{\beta  - 2}}} \right]}} \le 0.
\end{equation}\end{footnotesize}
Then, we can determine the upper bound for (\ref{eq_intermediate_bound}) as follows:
\begin{footnotesize}\begin{equation}
\label{eq:upper_bound}
\begin{aligned}
 &- \frac{{2a}}{{{{\left( {1 - \bar h} \right)}^3}}} + \frac{{9{{\left( {\bar h - \frac{1}{2}} \right)}^4}{{\left[ {\left( {\beta  + 1} \right){\gamma ^\beta } - \beta {\gamma ^{\beta  - 1}}} \right]}^2}}}{{2{{\left( {\bar h - \frac{1}{2}} \right)}^3}\left[ {\beta \left( {\beta  + 1} \right){\gamma ^{\beta  - 1}} - \beta \left( {\beta  - 1} \right){\gamma ^{\beta  - 2}}} \right]}}\\
 =&  - \frac{{2a}}{{{{\left( {1 - \bar h} \right)}^3}}} + \frac{{9\left( {\bar h - \frac{1}{2}} \right){{\left[ {\left( {\beta  + 1} \right)\gamma  - \beta } \right]}^2}{\gamma ^\beta }}}{{2\left[ {\beta \left( {\beta  + 1} \right)\gamma  - \beta \left( {\beta  - 1} \right)} \right]}}\\
 <&  - \frac{{2a}}{{{{\left( {1 - \bar h} \right)}^3}}} + \frac{{9\left( {\bar h - \frac{1}{2}} \right){{\left( {\beta  + 1} \right)}^2}{\gamma ^{\beta  + 2}}}}{{2\left[ {\beta \left( {\beta  + 1} \right)\gamma  - \beta \left( {\beta  - 1} \right)} \right]}}\\
 <&  - \frac{{2a}}{{{{\left( {1 - \bar h} \right)}^3}}} + \frac{{9\left( {\bar h - \frac{1}{2}} \right){{\left( {\beta  + 1} \right)}^2}{\gamma ^{\beta  + 2}}}}{{2\left[ {\beta \left( {\beta  + 1} \right)\gamma  - \beta \left( {\beta  - 1} \right)\gamma } \right]}}\\
 =&  - \frac{{2a}}{{{{\left( {1 - \bar h} \right)}^3}}} + \frac{{9\left( {\bar h - \frac{1}{2}} \right){{\left( {\beta  + 1} \right)}^2}{\gamma ^{\beta  + 2}}}}{{2\beta \gamma \left( {\beta  + 1 - \beta  + 1} \right)}}\\
 =&  - \frac{{2a}}{{{{\left( {1 - \bar h} \right)}^3}}} + \frac{{9\left( {\bar h - \frac{1}{2}} \right){{\left( {\beta  + 1} \right)}^2}{\gamma ^{\beta  + 1}}}}{{4\beta }}
 \end{aligned}
 \end{equation}\end{footnotesize}

Since \begin{footnotesize}${\bar{h}} \in \left[\frac{1}{2}, 1\right)$\end{footnotesize}, the last equation in~(\ref{eq:upper_bound}), i.e., \begin{footnotesize}$- \frac{{2a}}{{{{\left( {1 - \bar h} \right)}^3}}} + \frac{{9\left( {\bar h - \frac{1}{2}} \right){{\left( {\beta  + 1} \right)}^2}{\gamma ^{\beta  + 1}}}}{{4\beta }}$\end{footnotesize}, is smaller than \begin{footnotesize}$  - \frac{{2a}}{{{{\left( {1 - \frac{1}{2}} \right)}^3}}} + \frac{{9\left( {1 - \frac{1}{2}} \right){{\left( {\beta  + 1} \right)}^2}{\gamma ^{\beta  + 1}}}}{{4\beta }}=  - 16a + \frac{{9{{\left( {\beta  + 1} \right)}^2}{{\left( {{\gamma ^u}} \right)}^{\beta  + 1}}}}{{8\beta }}$\end{footnotesize}. \begin{footnotesize}$- 16a + \frac{{9{{\left( {\beta  + 1} \right)}^2}{{\left( {{\gamma ^u}} \right)}^{\beta  + 1}}}}{{8\beta }} < 0$\end{footnotesize} is satisfied according to Theorem~\ref{Th:uniqueness_for_leader_game} and hence the matrix in~(\ref{Eq:congruent_matrix_J_2}) has zero and negative definite blocks on its diagonal positions and zero blocks on the off-diagonal positions. This means that the matrix~(\ref{Eq:congruent_matrix_J_2}) is a negative semi-definite matrix. Since \begin{footnotesize}${\bf{J}}_2$\end{footnotesize} is congruent with the matrix~(\ref{Eq:congruent_matrix_J_2}), \begin{footnotesize}${\bf{J}}_2$\end{footnotesize} is also a negative semidefinite matrix.

To summarize, as the sum of a negative definite matrix \begin{footnotesize}${\bf{J}}_1$\end{footnotesize} and a negative semidefinite matrix \begin{footnotesize}${\bf{J}}_2$\end{footnotesize}, the Jacobian matrix in~(\ref{Eq:jacobian_matrix_leader_game}), i.e., \begin{footnotesize}${\bf{J}}$\end{footnotesize}, is a negative definite matrix. Thus, the NE in the leader-level noncooperative subgame \begin{footnotesize}${\cal{G}}_{\rm{L}}$\end{footnotesize} is unique and the Stackbelberg equilibrium is unique. The proof is completed.
\end{proof}

\end{document}